\documentclass[twocolumn,amsthm]{autart}

\usepackage[utf8]{inputenc}
\usepackage{comment}
\usepackage{algorithm}

\usepackage{algpseudocode}

\newcommand{\E}{\mathbb{E\,}}

\usepackage{graphicx}
\usepackage[version=4]{mhchem}
\usepackage{siunitx}
\usepackage{longtable,tabularx}
\usepackage{booktabs}
\usepackage{pifont}
\usepackage{graphics} 
\usepackage{float}
\usepackage{epsfig} 
\usepackage{times} 
\usepackage{amsmath} 
\usepackage{amssymb}  
\usepackage{mathtools}
\usepackage{color}
\usepackage[dvipsnames]{xcolor}
\usepackage{hyperref}

\newtheorem{assumption}{Assumption}

\newtheorem{theorem}{Theorem}
\newtheorem{proposition}{Proposition}
\newtheorem{corollary}{Corollary}

\usepackage{bm}
\usepackage{times}
\usepackage{soul}

\begin{document}

\begin{frontmatter}

\title{Control and Estimation Co-Design via Envelope-Theorem Gradients}

\author{Mohammad S. Ramadan}\ead{mramadan@anl.gov},
\author{Philip Dinenis}\ead{pdinenis@anl.gov},
\author{Mihai Anitescu}\ead{anitescu@mcs.anl.gov}

\address{Mathematics and Computer Science Division, Argonne National
Laboratory, Lemont, IL 60439, USA}


\begin{abstract}
Instead of the sequential plant--control--estimator design pipeline, which is in general not optimal,
we propose the Control and Estimation Co-Design (ContEst) framework, which poses the entire system design as a single two-stage problem solved by first-order methods. A key
enabler is that the gradient of each inner control/estimation optimal cost with
respect to the design parameters is available from the dual variables
the inner solver already returns, by the envelope theorem: exactly under mild
regularity (a unique inner minimizer), and as a subgradient under more relaxed conditions. That is, no differentiation
through the optimizer or its optimality conditions is needed. We first present the
framework in a linear quadratic Gaussian (LQG) regime, where the control and
estimation costs become two semidefinite programs (SDPs) whose design gradients are
read from the duals of the Lyapunov constraints. We extend the results to robust ($H_\infty$)
formulation to cover worst-case control and filtering problems. We also present an extension to constrained and nonlinear systems, using an extended Kalman filter (eKF) and a local linearization of the information-state dynamics to yield an approximate but convex model predictive control (MPC) problem whose design gradients are read directly from the costates. We apply our methods to a variety of realistic design examples from different fields and report significant design improvements.
\end{abstract}

\end{frontmatter}

\section{Introduction}\label{sec:intro}

Engineering systems are still, overwhelmingly, designed sequentially. A
plant is first conceived by some set of disciplines, the sensing and
instrumentation are specified next, and the controller and estimator are
designed last, once the plant, actuation and the measurement chain is frozen
\cite{garcia2019control}. Each stage restricts the next, so decisions made
early on, before any closed-loop consideration, can bound the performance
achievable at the end. This division of design effort is convenient, but
it is in general not optimal: the sequential solution (if exists) is one feasible
point of a larger joint problem and need not be an optimal or even a stationary point of it \cite{fathy2001coupling}.

Control co-design (CCD) addresses this on the control side, optimizing the
plant and the controller together~\cite{garcia2019control,allison2014multidisciplinary}. The same idea appeared in
process systems engineering as integrated design and control \cite{sakizlis2004recent,sharifzadeh2013integration} and in aerospace as
simultaneous structure/control design \cite{onoda1987approach,hale1985optimal}. What these literatures share is a
plant-controller scope: the sensing configuration and the estimator design are
typically fixed, or postponed to later steps. Yet for any
system without full and noiseless measurements (i.e. almost every real
system), achievable control performance is limited by the sensing and estimation quality.

This paper proposes the Control and Estimation Co-Design (ContEst) framework, in which the design parameters governing actuation and sensing/estimation are placed in a single two-stage optimization whose second-stage objective is the stochastic optimal control (dual control) cost. Because the system is
output feedback, the cost depends on the system not through the unobserved state but through its information state (commonly the filtering density
\cite{kumar2015stochastic}), whose evolution couples the output
equation and the state dynamics; optimizing through it therefore improves the controller and the estimator jointly. Mathematically, we solve the first stage with first-order methods (e.g.\ BFGS) driven by the stochastic optimal control gradient, which the envelope theorem supplies from the inner solver's dual variables (exactly under a unique inner minimizer, and as a subgradient otherwise). Across our studies this yields significant design cost reductions in which the sensing axis, absent from control-only co-design, is a major and sometimes dominant contributor.

\emph{Contributions:}
\begin{itemize}
    \item We formulate the control and estimation co-design as a single two-stage problem: a design first stage over a stochastic optimal control problem (Section~\ref{sec:formulation}).
    \item We show that the design gradients are obtained exactly from the dual variables, already returned by the solver. That is, no implicit differentiation of the Karush-Kuhn-Tucker (KKT) system is needed. Therefore, the gradient computation comes with no (or trivial) cost (Section~\ref{sec:envelope}).
    \item We provide regularity conditions (a unique inner minimizer and strict complementarity) under which the envelope gradient is exact, and show that when they fail the same returned dual still yields a valid subgradient, so the method degrades, but does not fail (Section~\ref{sec:envelope}; Appendices~\ref{app:lqgreg} and~\ref{app:subgrad}).
    \item Casting the inner problems as SDPs buys modeling flexibility such as: variety of convex design regularizers, quadratic stability over a polytopic (convex-hull) uncertainty set, and structural/sparsity constraints on the covariances or the control and estimator gains (Sections~\ref{sec:lqg} and~\ref{sec:hinf}).
    \item As an approximate extension to constrained nonlinear systems, a regime built around the linearization of eKF dynamics is used to construct a convex model predictive control (MPC) problem, with approximate gradients read from the
    dynamics costates (adjoint variables) (Section~\ref{sec:ekf}).
    \item Sensor allocation is absorbed directly into the formulation, either as a mixed-integer selection or as an $\ell_1$ relaxation promoting sparse sensing (Remark~\ref{rem:sensor}; applied in Section~\ref{sec:sim:mtdc}).
    \item We validate the framework on four realistic co-design studies across spacecraft, process, and power systems, with diverse objectives and modeling techniques (Section~\ref{sec:numerical}).
\end{itemize}

\section{Related work}\label{sec:related}

\emph{Control co-design and integrated design.}
\cite{fathy2001coupling} established that plant and controller optimization are
coupled and cannot be separated without loss; \cite{allison2014multidisciplinary}
cast dynamic-system CCD as multidisciplinary design optimization (MDO); and
\cite{herber2019nested} systematized the nested and simultaneous solution
strategies. In all of these the estimator design and sensing layers are typically assumed fixed.

\emph{Structural co-design and information architecture.}
\cite{grigoriadis1998integrate} alternated a control linear matrix inequality (LMI) with a plant parameter design. Closest to the present work, \cite{li2008integrating} treated the information architecture (sensor/actuator gains) as decision variables jointly with control or estimation via LMIs, and \cite{saraf2017h2,das2017sparse} pushed this toward sparse $H_2$ sensing with guaranteed estimation-error bounds. These methods alternate LMIs over gains for a fixed plant and cost class, whereas ContEst places actuation, sensing, and the estimator itself in one stochastic objective and obtains the joint design gradient exactly rather than relying on alternation.

\emph{Dual-based sensitivity for plant/controller co-design.}
The one line of work that shares our gradient mechanism is
\cite{moten2016gradient}, which solves plant-and-controller co-design as an outer
BFGS loop over plant parameters wrapped around an inner convex $H_\infty$ LMI, and
reads the gradient of the inner optimal value from the returned SDP duals via
parametric sensitivity analysis \cite{shapiro1997differentiability}. That the
optimal value of a parametric conic program is differentiated by its dual is
therefore not what we claim as new. ContEst differs in three respects. First, in
scope: their design vector enters the plant only, the output map $C_2$ is fixed,
and the objective is a deterministic worst-case transfer-function norm in which
estimation error cannot appear; ContEst places actuation, sensing, and the
estimator in one stochastic objective and reads the design gradient from the
control \emph{and} estimation LMI duals alike, across the $H_2$, $H_\infty$, and
eKF--MPC variants. Second, in rigor: the dual identity is invoked there without
uniqueness or nondegeneracy hypotheses, whereas we state the conditions (A1)--(A3)
under which it is exact and show in Section~\ref{sec:envelope} that the returned
dual remains a valid Clarke subgradient when they fail---which matters precisely
for a minimal-$\gamma$ $H_\infty$ inner problem, where the optimal dual is in
general set-valued. Third, in modeling reach: the full-order projection-lemma
parameterization they use presupposes unstructured Lyapunov certificates, and so
cannot express the structured gain patterns, sparsity-promoting regularizers, and
budgets coupling sensing to actuation that motivate our SDP path.

\emph{Sensor selection and LQG sensing co-design.}
A large literature selects sensors (or actuators) for estimation via
convex relaxation \cite{joshi2009sensor}, submodularity and greedy algorithms
\cite{tzoumas2016sensor,zhang2017sensor}, or sparsity-promoting selection
\cite{dhingra2014admm}. These covariance-based criteria are the control-theoretic
face of optimal experimental design (our estimation cost
$\operatorname{tr}(Q\Sigma^\varepsilon)$ is of A-optimal type, $\Sigma^\varepsilon$ is the state estimation error covariance), and the
identification community's least-costly experiment design
\cite{gevers2005identification} is a sibling of our noise level
co-design (Section~\ref{sec:weights}).
Closest to us, \cite{tzoumas2021lqg} co-design sensing with LQG control through
the separation structure (see also \cite{siami2021separation,zardini2021codesign});
but these choose \emph{which} sensors to deploy from a discrete menu, strictly in the
LQG regime.

\emph{Differentiable optimization and the envelope theorem.}
Recent work differentiates through convex control programs by implicitly differentiating the KKT system \cite{amos2017optnet,agrawal2019differentiable}. As we stress in
Section~\ref{sec:envelope}, for the optimal value (not the optimal argument) this machinery is unnecessary. Implicit differentiation of the KKT system solves an $m\times m$ linear system per design
direction, with $m=O(r_x^2)$ for the covariance SDPs ($r_x$ is the state-space dimension), i.e.\ $O(r_x^6)$ per
directional derivative, and it differentiates the minimizer, which ContEst
never needs. The envelope route reads the already-computed dual, therefore requires $O(r_x^2)$ for the whole gradient.

\emph{Parametric LQ and Riccati-based design gradients.}
The design gradient we recover is classical: the two-gramian formula dates to the
constant output feedback problem of \cite{levine1970determination} and the
parametric linear quadratic (LQ) machinery of \cite{makila1987computational,rautert1997computational}. On the estimation side \cite{belabbas2016geometric} differentiates the filter
Riccati, now a generic AD primitive \cite{kao2020autodiff} exploited by
adjoint-based co-design \cite{he2026efficient}. For unconstrained and plain linear quadratic regulator (LQR) inner problem, these
Riccati/adjoint routes are computationally cheaper than the SDPs we use, and ContEst can revert to the Riccati path when possible. Our
contribution is not the gradient in isolation but its unification, via the envelope theorem, across the control and estimation halves and their
$H_2$/$H_\infty$/eKF--MPC variants, and retention of the SDP path where it buys modeling power the Riccati route cannot express: quadratic stability over a convex uncertainty set, structural/sparsity
constraints, and others. Moreover, exploiting structure (diagonal or masked covariances) makes
the conic route cheaper in cost.

\emph{Positioning of the paper.}
Table~\ref{tab:positioning} gives a summary comparison of other methods per capability offered by ContEst (other methods may provide other benefits in other contexts). Nevertheless, ContEst need not operate in isolation; its only external output is a design gradient, it can be composed with the methods it is compared against rather than replacing them.

\begin{table*}[t]
\centering
\caption{Positioning of ContEst with respect to the nearest literature. Each column is a prior line of work, and each row is a capability; \checkmark\ means that line of work delivers that capability and \ding{55} means it cannot (or it does not natively).``Estimator
itself designed'' means tuning a dynamic estimator (e.g.\ bandwidth 
or noise level), not only sensor selection.}
\label{tab:positioning}
\small
\resizebox{\textwidth}{!}{%
\begin{tabular}{lcccccc}
\toprule
capability & CCD/MDO & sensor & LQG sensing & info-architecture
& diff.\ optimization & \textbf{ContEst}\\
 & \cite{garcia2019control,herber2019nested} & selection
 \cite{joshi2009sensor,tzoumas2016sensor} & co-design \cite{tzoumas2021lqg}
 & LMIs \cite{li2008integrating}
 & layers \cite{amos2017optnet,agrawal2019differentiable} & \\
\midrule
plant-actuation co-designed & \checkmark & \ding{55} & \ding{55} & \checkmark & \ding{55} & \checkmark\\
plant-sensing (output map) co-designed & \ding{55} & \checkmark & \checkmark~(menu selection) & \checkmark & \ding{55} & \checkmark~(menu selection or $\ell_1$)\\
estimator itself designed & \ding{55} & \ding{55} & \ding{55} & \ding{55} & \ding{55} & \checkmark\\
closed-loop control in objective & \checkmark & \ding{55} & \checkmark & \checkmark & \checkmark & \checkmark\\
nonlinear regime & \checkmark & \ding{55} & \ding{55} & \ding{55} & \checkmark & \checkmark~(eKF-MPC surrogate)\\
gradient w.r.t.\ design & often FD & \ding{55} & \ding{55} & $\approx$ LMI alternation & needs KKT diff. & exact from duals\\
\bottomrule
\end{tabular}}
\end{table*}

Throughout, $\| x \|_Q^2 = x^\top Q x$, $\mathbb{E}$ denotes expectation,
$\operatorname{tr}(\cdot)$ the trace, and $M \succeq 0$ ($M \succ 0$) positive
(semi)definiteness. For a signal $s$, $s_{i:j} = (s_i, s_{i+1}, \dots, s_j)$ denotes
the sequence from time $i$ to $j$. A superscript ${}^\star$ marks an optimal value.

\emph{Outline:} Section~\ref{sec:formulation} states the ContEst problem as a
two-stage program over the actuation parameters $\theta_f$ and the
sensing/estimation parameters $\theta_h$, with the stochastic optimal control
cost in the inner stage. Section~\ref{sec:lqg} specializes the inner stage to
the LQG case, where control and estimation are a pair of SDPs, and shows how
sensor precision enters as a noise-level design. Section~\ref{sec:envelope}
develops the paper's main tool: exact design gradients obtained from the inner
solver's dual variables via the envelope theorem, together with the regularity
conditions under which they are exact and the subgradient guarantee when those
conditions fail. Section~\ref{sec:hinf} gives the $H_\infty$ analog through the
bounded-real lemma, and Section~\ref{sec:ekf} extends the framework to
constrained nonlinear systems via an eKF--MPC surrogate whose approximate
gradients come from the dynamics costates. Section~\ref{sec:numerical} reports
four co-design studies --- spacecraft attitude determination and control, a
binary distillation column, PLL/DSE tuning in a low-inertia grid, and sparse
sensing in a multi-terminal HVDC ring --- followed in
Section~\ref{sec:sim:complexity} by a complexity analysis of the two inner
problems. Section~\ref{sec:conclusion} concludes. Proofs and the information-state
derivations are collected in the appendices.

\section{Problem formulation}\label{sec:formulation}

Consider the parametrized stochastic system
\begin{subequations}\label{eq:ss}
\begin{align}
x_{k+1} &= f_\theta(x_k, u_k) + w_k, \label{eq:ss-state}\\
y_k     &= h_\theta(x_k) + v_k, \label{eq:ss-output}
\end{align}
\end{subequations}
with state $x_k \in \mathbb{R}^{r_x}$, input $u_k \in \mathbb{R}^{r_u}$, and
output $y_k \in \mathbb{R}^{r_y}$. The maps $f_\theta, h_\theta$ are parametrized by $\theta \in \mathbb{R}^{r_\theta}$, representing plant parameters, actuator gains, and sensor tuning/placement. The state dynamics $f_\theta$ is locally bounded in its arguments for all $\theta$. The disturbances $w_k, v_k$ are i.i.d.,
mutually independent and independent of $x_0 \sim p_0$ (the prior, before any
measurement including $y_0$), with covariances $W,V$, respectively, which are possibly parametrized by $\theta$ as well.

The goal is to jointly optimize design, control, and estimation altogether. The
design cost $J_{\mathrm{des}}(\theta)$ and design set $\theta \in \Theta$
encode the price, manufacturability, placement difficulty or some mechanical, thermal, structural, and other considerations. The
control and estimation costs are embedded in the stochastic optimal control problem
\begin{equation} \label{eq:Jstoch_N}
\begin{aligned}
J_{\mathrm{sc}}^\star(\mathcal Z_0;\theta) &= \min_{u_{0:N-1}} J_{\mathrm{sc}}(\mathcal Z_0,u_{0:N-1};\theta),\\
&=\min_{u_{0:N-1}} \frac{1}{N}\,
\mathbb{E}\big\{ \sum_{k=0}^{N-1} \ell(x_k,u_k) + \ell^N(x_N) \big\},
\end{aligned}
\end{equation}
for finite-horizon problems, or for infinite-horizon as in 
\begin{equation}\label{eq:Jstoch_infty}
J_{\mathrm{sc}}^\star(\mathcal Z_0;\theta) = \min_{\kappa} \lim_{k \to \infty} \mathbb{E} \, \ell(x_k,\kappa(\mathcal Z_k)),
\end{equation}
with the $\ell, \ell^N$ are stage and terminal costs are typically
\begin{equation*}
   \ell(x_k,u_k) = \|x_k \|_Q ^2 + \| u_k\|_R^2, \quad \ell^N(x_N) = \| x_N \|_Q^2,
\end{equation*}
and the information available at the end of time $k$ is $\mathcal Z_k = (p_0, y_{0:k}, u_{0:k})$ (at time $k$, $y_k$ is measured just after $u_k$ is applied, hence a causal $u_k$ is a function of $\mathcal{Z}_{k-1}$ at most).

The above minimizations are possibly subject to state and input constraints and the expectation is over $(x_0, w_{0:N-1}, v_{0:N})$.

\begin{prob}[ContEst]\footnote{Problem~\ref{prob:two-stage} is a two-stage problem, a special form of a bilevel problem; the outer optimization contains the value of the lower-level problem, not its argmin. Solving a bilevel problem directly is hard since it requires implicitly differentiating the inner optimality conditions. Instead, in the two-stage case, the envelope theorem supplies the gradient
directly from the duals the inner solver already returns.}\label{prob:two-stage}
\begin{align*}
&\text{Design (first stage):}\quad \min_{\theta \in \Theta}\
J(\theta) = J_{\mathrm{des}}(\theta) + J_{\mathrm{sc}}^\star(\mathcal Z_0;\theta),\\
&\text{Control (second stage):}\quad J_{\mathrm{sc}}^\star(\mathcal Z_0;\theta)\ \text{as in
\eqref{eq:Jstoch_N} or \eqref{eq:Jstoch_infty}}.
\end{align*}
\end{prob}

\begin{rem}[Estimation co-design]
    Equations \eqref{eq:Jstoch_N} and \eqref{eq:Jstoch_infty}, although appear to be optimizing over control only, they implicitly require optimal estimation as well. The control policies are feedback in the information state (state filtered density), which require careful filter design for their reduction, approximation and accurate tracking \cite{kumar2015stochastic}.
\end{rem}

\begin{rem}[Operating scenarios]\label{rem:scenarios}
Problem~\ref{prob:two-stage} is stated for a single prior $p_0$. Multiple operating scenarios: initial densities, tracking trajectories/set-point, disturbance statistics, can be accommodated by replacing the second-stage value with the, say, a weighted average of the cost over the different scenarios.
\end{rem}

Note that the second stage is a stochastic (dual) control problem whose sufficient
statistic is the information state: the filtered density in this paper. The exact reduction of this problem into two LQG SDPs in the next section, and approximate reduction to eKF-MPC in a later section are detailed in Appendix~\ref{app:infostate}.

\section{The LQG case: LQR and Kalman SDPs}\label{sec:lqg}

We start with a linear (or linearized) plant
$$x_{k+1}=A_\theta x_k+B_\theta u_k+w_k, \quad y_k=C_\theta x_k+v_k,$$ 
with control input $u_k = \kappa (\mathcal{Z}_{k-1}) = K x_{k \mid k-1}$, the certainty equivalent state-feedback law, and $x_{k+1 \mid k}$ is the state of the observer at time $k+1$, given all the information up to the beginning of time $k+1$, i.e. $\mathcal{Z}_{k}$,
\begin{align*}
     x_{k+1 \mid k} = (A_\theta + B_\theta K) x_{k \mid k-1} + G(y_{k} - C_\theta x_{k \mid k-1}),
\end{align*}
with the observer (possibly Kalman) gain $G$. This results in the estimation error $\varepsilon_k = x_k - x_{k \mid k-1}$ with dynamics described by
\begin{align*}
    \varepsilon_{k+1} = (A_\theta - G C_\theta ) \varepsilon_{k} + w_k - G v_k.
\end{align*}
The second stage of Problem~\ref{prob:two-stage} is an LQG problem whose cost splits, by
the separation principle, into control and estimation parts
\begin{equation} \label{eq:linear stoch cost}
\begin{aligned}
&J_{\mathrm{sc}}=J_{\mathrm{est}}(\theta)+J_{\mathrm{cont}}(\theta),\\
&J_{\mathrm{est}}=\operatorname{tr}(Q\Sigma^\varepsilon), \quad
J_{\mathrm{cont}}=\operatorname{tr}(Q\Sigma)+\operatorname{tr}(RK\Sigma K^\top)
\end{aligned}
\end{equation}
where $\Sigma=\lim_{k \to \infty}\mathbf{Cov}( x_{k \mid k-1})$ is the
steady-state covariance of the state of the observer, and $\Sigma^\varepsilon=\lim_{k \to \infty}\mathbf{Cov}(\varepsilon_k)$ is the steady-state
estimation-error covariance.

For a fixed $\theta$, the two parts in \eqref{eq:linear stoch cost}, $J_{\mathrm{est}}$ and $J_{\mathrm{cont}}$, depend on disjoint variables $G$ and $K$, respectively, and therefore are minimized
independently \cite{tzoumas2021lqg}, but at the design level $\theta$ couples
$A_\theta,B_\theta,C_\theta$ and the noise covariances, and hence, $\theta$ also couples both the control
gain $K$ and the observer gain $G$. Since $\theta$ enters the system through the system matrices $A_\theta,B_\theta,C_\theta$ and possibly the covariance matrices $W,V$, we seek next to first find the jacobians of the optimal $J_{\mathrm{est}}$ and $J_{\mathrm{cont}}$ with respect to these matrices, then use the chain rule to evaluate the gradient with respect to $\theta$ (that is why we need the differentiability of these matrices with respect to $\theta$ in Assumption~\ref{as:env}-\emph{(A4)}).

\emph{Control SDP.} Following the linearizing change of
variables $L=K\Sigma$ (with the epigraph slack $Z_0$ upper-bounding the input
second moment $K\Sigma K^\top$ through the first LMI), we have \cite{boyd1993control}
\begin{equation}\label{eq:Jcont}
\begin{aligned}
&J_{\mathrm{cont}}^\star(\theta) = \min_{\Sigma,L,Z_0}\
\operatorname{tr}(Q\Sigma) + \operatorname{tr}(R Z_0),\quad \text{s.t. }\\
&
\begin{bmatrix} Z_0 & L \\ L^\top & \Sigma \end{bmatrix}
\succeq 0, \,\,\,\begin{bmatrix}
\Sigma - W & A_\theta\Sigma + B_\theta L \\
(A_\theta\Sigma + B_\theta L)^\top & \Sigma
\end{bmatrix} \succeq 0,
\end{aligned}
\end{equation}
with optimal gain $K^\star = L^\star (\Sigma^\star)^{-1}$.

\emph{Estimation SDP.} Dually, and moving from the controllability gramian $\Sigma^\varepsilon$ to the observability one $P^\varepsilon$ that solves a
Lyapunov inequality, which after the change of variables $F = P^\varepsilon G$
(with slack $Z_1$) and the LMI relaxation, gives \cite{boyd1993control}
\begin{equation}\label{eq:Jest}
\begin{aligned}
&J_{\mathrm{est}}^\star(\theta) = \min_{P^\varepsilon,F,Z_1}\
\operatorname{tr}(W P^\varepsilon) + \operatorname{tr}(V Z_1),\quad \text{s.t. }\\
&\begin{bmatrix} Z_1 & F^\top \\ F & P^\varepsilon \end{bmatrix}
 \succeq 0, \,\,\, \begin{bmatrix}
P^\varepsilon - Q & A_\theta^\top P^\varepsilon - C_\theta^\top F^\top \\
P^\varepsilon A_\theta - F C_\theta & P^\varepsilon
\end{bmatrix} \succeq 0,
\end{aligned}
\end{equation}
and $G ^ \star = (P^{\varepsilon\, \star})^{-1}F^\star$, which is dual in form to \eqref{eq:Jcont} and exposing the sensing operator $C_\theta$,
i.e.\ the estimation half of the co-design \cite{li2008integrating}.

\emph{Design gradients} Let us first define the general form Lagrangian we will be using across this paper\footnote{We abuse the notation by redefining $\mathfrak L$ several times across the paper, the particular definition we refer to in the text is inferred from context.}
\begin{align} \label{eq:lagrangian general}
    &\mathfrak L(p, \{\lambda_i\}_{i=1}^{r_{\lambda}}, \{S_i \}_{i=1}^{r_{S}}; \theta) = \nonumber\\
    &\hskip 10mm V(p; \theta) - \sum_{i=1}^{r_{\lambda}}\left \langle \lambda_i, m_i \right \rangle - \sum_{i=1}^{r_{S}} \left \langle S_i, M_i\right \rangle,
\end{align}
where $p$ is the primal variable, $V$ is the cost function, $\lambda_i's$ and $S_i's$ correspond to the equality constraints $m_i = 0$ and LMIs $M_i \succeq 0$, respectively. For the particular case of \eqref{eq:Jcont} (or similarly in form for the estimation problem \eqref{eq:Jest})
\begin{align} \label{eq:lagrangian lqr}
    &\mathfrak L(\Sigma, L , Z_0, S_1, S_2; \theta) = \nonumber\\
    &\hskip 10mm \operatorname{tr}(Q\Sigma) + \operatorname{tr}(R Z_0) - \left \langle S_1, M_1 \right \rangle - \left \langle S_2, M_2\right \rangle,
\end{align}
where $M_1$ and $M_2$ are the left-hand-side matrices of the epigraph and Lyapunov LMIs in \eqref{eq:Jcont}, and $S_1, S_2 \succeq 0$ are their dual variables, all respectively. The envelop theorem under some regularity conditions, as will be detailed in Section~\ref{sec:envelope}, states that the gradient of $\nabla_\theta J_{cont}^\star$ equals to $\nabla_\theta L$ with the primal and dual variables fixed to their optimal values
\begin{align*}
    \nabla_\theta J_{cont}^\star &= \nabla_\theta \mathfrak L(\Sigma, L^\star , Z_0^\star, S^\star_1, S^\star_2; \theta),\\
    &= - \nabla_\theta \left \langle S^\star_2, M_2\right \rangle,
\end{align*}
since $\theta$ enters only through $M_2$,
\begin{align*}
    &= - \nabla_\theta \operatorname{tr} \Big (
\begin{bmatrix} 
    S_{2,[1,1]}^\star & S_{2,[1,2]}^\star \\ 
    S_{2,[1,2]}^{\star\,\top} & S_{2,[2,2]}^\star
\end{bmatrix} \times \\
&\hskip 20mm
\begin{bmatrix}
\Sigma^\star - W & A_\theta\Sigma^\star + B_\theta L^\star \\
(A_\theta\Sigma^\star + B_\theta L^\star)^\top & \Sigma^\star
\end{bmatrix}
    \Big ),
\end{align*}
where $S_{2,[i,j]}^\star$ denotes the corresponding-in-dimension block $(i,j)$ of $S_2^\star$. Towards using the chain rule, we first find the jacobians with respect to the system matrices
\begin{equation}\label{eq:gradAB}
\frac{\partial J_{\mathrm{cont}}^\star}{\partial A_\theta} = -2\, S_{2,[1,2]}^\star\Sigma^\star,
\quad
\frac{\partial J_{\mathrm{cont}}^\star}{\partial B_\theta} = -2\, S_{2,[1,2]}^\star(L^\star)^\top,
\end{equation}
and similarly for the estimation case in \eqref{eq:Jest}, with
$\tilde S \succeq 0$ defined as the corresponding Lyapunov inequality dual variable in a Lagrangian similar in form to \eqref{eq:lagrangian lqr},
\begin{equation}\label{eq:gradAC}
\frac{\partial J_{\mathrm{est}}^\star}{\partial A_\theta} = - 2\, P^{\varepsilon\star} \tilde S^{\,\star\,\top}_{2,[1,2]},
\quad
\frac{\partial J_{\mathrm{est}}^\star}{\partial C_\theta} = 2\, F^{\star\,\top} \tilde S^{\,\star\,\top}_{2,[1,2]}.
\end{equation}

Notice the forms of the jacobians in \eqref{eq:gradAB} and \eqref{eq:gradAC}, explained by the facts: (i) the control and estimation problems are transposed duals, swapping $(A_\theta, B_\theta, Q, R, W)$ with $(A_\theta^\top, C_\theta^\top, W, V, Q)$, (ii) $\partial \cdot / \partial A_\theta^\top = (\partial \cdot / \partial A_\theta)^\top$, and (iii) $F$ swapped with $-L^\top$.

The design gradient follows by the chain rule
(Section~\ref{sec:envelope}),
\begin{align}\label{eq:chain}
\frac{\partial J_{\mathrm{cont}}^\star}{\partial \theta_i} &= \left \langle\frac{\partial J_{\mathrm{cont}}^\star}{\partial A_\theta},
\frac{\partial A_\theta}{\partial \theta_i}\right \rangle
+ \left \langle\frac{\partial J_{\mathrm{cont}}^\star}{\partial B_\theta},
\frac{\partial B_\theta}{\partial \theta_i}\right \rangle.
\end{align}
The full gradient vector is constructed from its components
\begin{align*}
    \nabla_\theta J_{\mathrm{sc}}^\star = \Big [\left(\frac{\partial J_{\mathrm{cont}}^\star}{\partial \theta_1} + \frac{\partial J_{\mathrm{est}}^\star}{\partial \theta_1}\right )\,\, \left (\frac{\partial J_{\mathrm{cont}}^\star}{\partial \theta_2} + \frac{\partial J_{\mathrm{est}}^\star}{\partial \theta_2} \right )\, \cdots  \Big]^\top.
\end{align*}

\subsection{Sensor precision: noise level co-design}\label{sec:weights}
The weights $Q, R$ of the quadratic cost and the noise covariances $W, V$ of the
process and measurement can also be parametrized by $\theta$. This matters in practice because they can represent the knobs a sensor or actuator budget buys.

Concretely, both LQG subproblems are affine in these data, so their sensitivities
are read off directly, with no LMI-dual partition required for the objective
terms. From the control SDP~\eqref{eq:Jcont},
\begin{equation}\label{eq:gradQRcont}
\frac{\partial J_{\mathrm{cont}}^\star}{\partial Q_\theta} = \Sigma^\star,\quad
\frac{\partial J_{\mathrm{cont}}^\star}{\partial R_\theta} = Z_0^\star,\quad
\frac{\partial J_{\mathrm{cont}}^\star}{\partial W_\theta} = \,S_{2,[1,1]}^\star,
\end{equation}
the first two from the objective $\operatorname{tr}(Q_\theta\Sigma)+\operatorname{tr}(R_\theta Z_0)$ and the
last from the $(1,1)$ block of $S^\star_2$. From the estimation SDP~\eqref{eq:Jest},
\begin{equation}\label{eq:gradQRest}
\frac{\partial J_{\mathrm{est}}^\star}{\partial W_\theta} = P^{\varepsilon\star},\quad
\frac{\partial J_{\mathrm{est}}^\star}{\partial V_\theta} = Z_1^\star,\quad
\frac{\partial J_{\mathrm{est}}^\star}{\partial Q_\theta} = \tilde S_{2,[1,1]}^{\star}.
\end{equation}
The total gradient is then the sum of all the contributions, added to \eqref{eq:chain} via chain rule.

\begin{rem}[Modeling flexibility]
Posing both halves as SDPs is not only convenient for the dual-based gradients:
it also opens the formulation to constraints a plain Riccati recursion cannot
express: LMI regularizers (e.g.\ trace or nuclear-norm penalties promoting sparse
actuation/sensing), budget and $\ell_1$ sensor-selection constraints on $\theta$,
or quadratic stability robustness constraints over a polytopic data set. Remark~\ref{rem:regexact} in the next section
shows every such convex ingredient is inherited by the exact-gradient machinery of
Section~\ref{sec:envelope}.
\end{rem}

\section{Exact design gradients via the envelope theorem}\label{sec:envelope}

The convex second-stage problem is a special case of the following conic form:
\begin{equation}\label{eq:generic}
J^\star(\theta) = \min_{z}\ J(z,\theta)\quad \text{s.t.}\quad g(z,\theta)=0,\ \
G(z,\theta) \succeq 0,
\end{equation}
with Lagrangian $\mathfrak L(z,\mu,S,\theta) = J(z,\theta) + \langle \mu, g(z,\theta)\rangle
- \langle S, G(z,\theta)\rangle$, and optimal primal-dual triple
$(z^\star,\mu^\star,S^\star)$. We collect the assumptions under which $J^\star$ is differentiable and its gradient is read from the primal-dual triple.

\begin{assumption}\label{as:env}
On a neighborhood $\Theta_0$ of the design point $\theta$:
\emph{(A.I)}~for each $\theta\in\Theta_0$, \eqref{eq:generic} is convex in $z$,
its data $(\phi,g,G)$ are jointly $C^1$ in $(z,\theta)$, and Slater's condition
holds, so strong duality holds with an attained primal-dual solution;
\emph{(A.II)}~the primal minimizer $z^\star(\theta)$ is unique;
\emph{(A.III)}~the primal--dual solution is \emph{nondegenerate} and strictly
complementary, so the dual $(\mu^\star,S^\star)$ is unique
\cite{alizadeh1997complementarity,shapiro1997differentiability};
\emph{(A.IV)}~$\theta$ enters \eqref{eq:generic} only through finitely many data
objects $D(\theta)$ (e.g.\ $A_\theta,B_\theta,C_\theta$) that are $C^1$
in $\theta$.
\end{assumption}

Assumptions (A.I)-(A.III) are the regularity under which the optimal value
of a conic (in particular semidefinite) program is differentiable and its
sensitivity is governed by the unique dual
\cite{bonnans2000perturbation,shapiro1997differentiability}: (A.I) establishes a no duality-gap, which is necessary to establish $J^\star = \mathfrak L(z^\star,\mu^\star,S^\star)$, (A.II) and (A.III) guarantee the optimal primal-dual uniqueness, hence the uniqueness of the gradient of the Lagrangian at the optimal primal-dual point, (A.IV) is a modeling convenience so the chain rule can be evaluated cheaply.

\begin{theorem}[Envelop theorem: exact design gradient from the inner dual]\label{prop:envelope}
Under Assumption~\ref{as:env}, $J^\star$ is differentiable at $\theta$ and
\begin{equation}\label{eq:envelope}
\nabla_\theta J^\star(\theta) = \nabla_\theta \mathfrak L(z^\star,\mu^\star,S^\star,\theta)
= \partial_\theta J + \langle \mu^\star, \partial_\theta g\rangle
- \langle S^\star, \partial_\theta G\rangle,
\end{equation}
all partials evaluated at $(z^\star,\theta)$ with $z^\star,\mu^\star,S^\star$
held fixed.
\end{theorem}
\begin{proof}
By (A.I)-(A.III) the solution map $\theta\mapsto(z^\star,\mu^\star,S^\star)$ is
single-valued and differentiable on $\Theta_0$
\cite{bonnans2000perturbation,shapiro1997differentiability}, and strong duality
with complementary slackness gives
$J^\star(\theta) = \mathfrak L(z^\star(\theta),\mu^\star(\theta),S^\star(\theta),\theta)$.
Differentiating and using inner stationarity $\nabla_z \mathfrak L = 0$ and the
stationarity of $\mathfrak L$ in $(\mu,S)$ at the solution,
$\nabla_\theta J^\star = \nabla_z \mathfrak L\,\partial_\theta z^\star
+ \nabla_{(\mu,S)}\mathfrak L\,\partial_\theta(\mu^\star,S^\star) + \partial_\theta \mathfrak L
= \partial_\theta \mathfrak L$ \cite{milgrom2002envelope,bonnans2000perturbation}.
\end{proof}

\begin{rem}\label{rem:regexact}
We can augment \eqref{eq:generic} with any finite family of additional convex,
conic-representable ingredients: regularizers $r_i(z,\theta)$ added to the objective
$J$, and/or constraints $G_i(z,\theta) \succeq 0$, all data $C^1$ in $(z,\theta)$. If the augmented program still
satisfies Assumption~\ref{as:env}, then $J^\star$ remains differentiable at $\theta$.
\end{rem}

\begin{rem}[Where exactness is lost]\label{rem:exactfail}
The recipe of Remark~\ref{rem:regexact} fails exactly when (A.III) does: at a
degenerate active face (rank-deficient active supports), or when a constraint
binds at its achievable floor (loss of strict complementarity: for instance
the minimal-$\gamma$ optimum of Section~\ref{sec:hinf}). The optimal dual is then
set-valued and the exact identity \eqref{eq:envelope} need not hold; the next
proposition shows the framework degrades gracefully rather than breaking.
\end{rem}

\begin{proposition}[Subgradient fallback under (A.I)]\label{prop:subgrad}
Suppose (A.I) and (A.IV) hold on a neighborhood of $\theta$ with a
bounded primal optimal set at $\theta$, but (A.II) or (A.III) fail, while at least one of the primal or the dual optimal solution at $\theta$ is still unique.
Then $J^\star$ is locally Lipschitz and directionally differentiable, and the
particular dual $(\mu^\star,S^\star)$ the inner solver returns constructs a subgradient, $\partial_\theta \mathfrak L(z^\star,\mu^\star,S^\star,\theta)\in\partial J^\star(\theta)$
(the Clarke subdifferential). If the dual is unique, $\partial_\theta L$ also supplies a descent direction and
first-order stationarity reads $0\in\partial J^\star(\theta)$.
\end{proposition}
\begin{proof}
Under (A.I) the value $J^\star$ is locally Lipschitz, as shown in Appendix~\ref{app:subgrad}, hence
its Clarke set $\partial J^\star(\theta)$ is nonempty, compact and convex. Uniqueness of the primal (resp.\ dual) solution yields a one-sided first-order bound on $J^\star$ at $\theta$ along every direction, which places
$\partial_\theta \mathfrak L(z^\star,\mu^\star,S^\star,\theta)$ in $\partial J^\star(\theta)$; when (A.II) and (A.III) both hold, the two bounds combine into the exact derivative of Theorem~\ref{prop:envelope}. Without either uniqueness the inclusion can fail (Appendix~\ref{app:subgrad}). 
\end{proof}

We now return to the validity of Assumption~\ref{as:env} and verify, for the LQG SDPs of Section~\ref{sec:lqg}, the regularity conditions required to satisfy this assumption.

\begin{proposition}[LQG regularity]\label{prop:lqgreg}
Fix $\theta$ and suppose $(A_\theta,B_\theta)$ is stabilizable, $(A_\theta,
C_\theta)$ is detectable, and $Q,R,W,V\succ0$, with all data $C^1$ in $\theta$.
Then for the control SDP~\eqref{eq:Jcont} and the estimation
SDP~\eqref{eq:Jest}: \emph{(a)}~Slater's condition holds; \emph{(b)}~the primal
minimizer $(\Sigma^\star, K^\star=L^\star(\Sigma^\star)^{-1})$ is unique (and similarly for the estimation SDP);
\emph{(c)}~the Lyapunov-LMI dual is unique and strictly complementary. Hence Assumption~\ref{as:env} holds
on a neighborhood $\Theta_0$ of $\theta$, and Theorem~\ref{prop:envelope} applies to both the control and estimation SDPs.
\end{proposition}

The proof is in Appendix~\ref{app:lqgreg}.
\hfill$\square$

Controllability and observability imply the stabilizability/detectability
hypotheses, so the practical reading of Proposition~\ref{prop:lqgreg} is:
controllable, observable plant with positive-definite weights $\Rightarrow$
envelop theorem can be used.

\section{The $H_\infty$ analog}\label{sec:hinf}

The LQG development of Section~\ref{sec:lqg} considers average performance
`against' stochastic white disturbances. When the disturbances are better modeled as
unknown-but-bounded $\ell_2$ signals and the designer cares about the
worst case, an $H_\infty$ approach can be used. The inner value is now the
squared worst-case closed-loop $\ell_2$-gain $\gamma^2(\theta)$ rather than a
stationary variance.

Write the linear plant with an explicit disturbance channel and a performance
output,
\begin{equation}\label{eq:hinf_plant}
x_{k+1} = A_\theta x_k + B_\theta u_k + E_\theta w_k,\quad
z_k = C_{z}\,x_k + D_{z}\,u_k,
\end{equation}
where $w_k\in\ell_2$ is the exogenous disturbance, $E_\theta$ its input map
(e.g.\ $E_\theta = W_\theta^{1/2}$, so the process noise level becomes a disturbance weight), and $z_k$ stacks the penalized
state and input, $C_z = [\,Q^{1/2}\;\,0_{r_x \times r_u}\,]^\top$, $D_z = [\,0_{r_u \times r_x}\;\,R^{1/2}\,]^\top$, so the
$H_\infty$ objective weighs exactly the same signals as in $J_{\mathrm{sc}}^\star$.

\emph{Control SDP.} By the discrete-time bounded-real lemma (BRL) with the standard
linearizing change of variables $Y = P^{-1}\succ 0$, $L = KY$
\cite{boyd1994linear,grigoriadis1998integrate}, the state-feedback $H_\infty$
synthesis problem is the SDP (with $\mu := \gamma^2$, $\epsilon > 0$ very small)
\begin{equation}\label{eq:Hinf_cont}
\begin{aligned}
&\gamma^2_{\mathrm{cont}}(\theta) = \min_{Y,L,\mu}\ \mu \\
&\text{s.t. }  Y \succeq \epsilon I,\ \
-\begin{bmatrix}
-Y & \star & \star & \star \\
0 & -\mu I & \star & \star \\
A_\theta Y + B_\theta L & E_\theta & -Y & \star \\
C_z Y + D_z L & 0 & 0 & -I
\end{bmatrix} \succeq 0,
\end{aligned}
\end{equation}
(the $\star$ blocks are inferred by symmetry) with optimal gain
$K^\star = L^\star (Y^\star)^{-1}$.

\emph{Design gradients from the duals.} With a Lagrangian similar in convention to that in \eqref{eq:lagrangian lqr} and dual variable $S_{\infty,2}$ corresponding to the second LMI in \eqref{eq:Hinf_cont}, the plant blocks
$A_\theta Y + B_\theta L$ sit in position $(3,1)$, and therefore the formulae for the jacobians
\begin{equation}\label{eq:gradHinf}
\frac{\partial \gamma^{\star\,2}_{\mathrm{cont}}}{\partial A_\theta} = 2\,S_{\infty,2,[3,1]}^\star\,Y^\star,
\quad
\frac{\partial \gamma^{\star\,2}_{\mathrm{cont}}}{\partial B_\theta} = 2\,S_{\infty,2,[3,1]}^\star\,(L^\star)^\top,
\end{equation}
with analogous expressions for $E_\theta$ (from the $(3,2)$ block $S_{\infty,2,[3,1]}^\star$),
and the design gradient assembled by the same chain rule~\eqref{eq:chain}.

\begin{corollary}[$H_\infty$ minimal-$\gamma$ design: a subgradient]\label{cor:hinf}
The expression \eqref{eq:gradHinf} for the minimal-$\gamma$ value
$\gamma^{\star\,2}_{\mathrm{cont}}(\theta)$ is, in general, only a subgradient: it
supplies a descent direction and the stationarity certificate
$0\in\partial\gamma^{\star\,2}_{\mathrm{cont}}$, but need not equal the gradient. Fixing $\gamma$ instead restores exactness, as we show next.
\end{corollary}
\begin{proof}
At the optimum \eqref{eq:Hinf_cont} drives $\mu$ down to the optimal
attenuation $\gamma^{\star\,2}_{\mathrm{cont}}$, where the BRL LMI loses rank: the $(2,2)$
block $-\mu I$ tightens against the Schur complement of the plant blocks, so
the active face is degenerate and strict complementarity (A.III) generically
fails. By Proposition~\ref{prop:subgrad} the value is then locally Lipschitz
and the BRL dual the solver returns is a subgradient of
$\gamma^{\star \, 2}_{\mathrm{cont}}$.
\end{proof}

\emph{Fixed-$\gamma$ SDP: exact gradients.} Fix the attenuation $\gamma$
strictly above the optimal level (the problem \eqref{eq:Hinf_cont}) can be solved for $\gamma^\star$, then $\gamma$ is picked larger with some gap). Rather than minimizing $\mu$, one then
minimizes the worst-case cost itself,
$\operatorname{tr}(XW_\theta)=\operatorname{tr}(Y^{-1}W_\theta)$, over the BRL LMI at
that fixed $\gamma$; with an epigraph slack $T$ and a Schur complement this is the SDP
\begin{equation}\label{eq:Hinf_fixed}
\begin{aligned}
V_{\mathrm{wc}}(\theta, \gamma)=&\min_{Y,L,T}\ \operatorname{tr}(T)\\
\text{s.t. } & Y\succeq \epsilon I,\quad
\begin{bmatrix} T & E_\theta\\ E_\theta & Y\end{bmatrix}\succeq0,\\
& - \begin{bmatrix}
-Y & \star & \star & \star \\
0 & -\gamma^2 I & \star & \star \\
A_\theta Y + B_\theta L & E_\theta & -Y & \star \\
C_z Y + D_z L & 0 & 0 & -I
\end{bmatrix} \preceq 0,
\end{aligned}
\end{equation}
whose last LMI is exactly that of \eqref{eq:Hinf_cont} with $\mu=\gamma^2$ held
constant, and whose Schur slack enforces
$T\succeq W_\theta^{1/2}Y^{-1}W_\theta^{1/2}$, so $\operatorname{tr}(T)$ upper-bounds and at the optimum equals $\operatorname{tr}(Y^{-1}W_\theta)=\operatorname{tr}(XW_\theta)$.

\begin{corollary}[$H_\infty$ fixed-$\gamma$ design: exact gradient]\label{cor:hinf_fixed}
Fix $\gamma$ strictly above the optimal attenuation. Then \eqref{eq:Hinf_fixed} is
strictly feasible with a nondegenerate, strictly complementary dual, so
Assumption~\ref{as:env} holds and, by Theorem~\ref{prop:envelope} (equivalently
Remark~\ref{rem:regexact}, the Schur slack being one added convex constraint),
$V_{\mathrm{wc}}(\theta, \gamma)$ is differentiable at $\theta$, with the design
gradient read from the dual of the third LMI $S^\star_{\gamma,3}$ of \eqref{eq:Hinf_fixed},
\begin{equation}\label{eq:gradHinf_fixed}
\frac{\partial V_{\mathrm{wc}}}{\partial A_\theta}=2\,S_{\gamma,3,[3,1]}^\star\,Y^\star,\quad
\frac{\partial V_{\mathrm{wc}}}{\partial B_\theta}=2\,S_{\gamma,3,[3,1]}^\star\,(L^\star)^\top,
\end{equation}
of the same form as the minimal-$\gamma$ expression \eqref{eq:gradHinf} but now exact gradients rather than mere subgradients.
\end{corollary}

\emph{Estimation SDP (fixed-$\gamma_e$): exact gradient.} Exactly as
\eqref{eq:Jest} is the transpose dual of \eqref{eq:Jcont} under
$(A,B,Q,R,W)\to(A^\top,C^\top,W,V,Q)$, fixing $\gamma_e$ strictly above the
optimal filter attenuation and applying the same swap to \eqref{eq:Hinf_fixed}
gives the estimation analog. That is, we also swap $(E_\theta, C_z, D_z) \to (Q^{1/2}, [\,W_\theta^{1/2}\;\,0_{r_x \times r_u}\,]^\top, [0_{r_u \times r_x}\;\,V_\theta^{1/2}]^\top)$. The optimal filter gain $G^\star$ is the transpose of that returned ($K^\star$) in the control case. Similarly, the jacobians w.r.t. $A_\theta$ and $C_\theta$ are the transpose of the form of the jacobians in \eqref{eq:Hinf_fixed} w.r.t. $A_\theta$ and $B_\theta$, respectively.

\emph{Scalable game-Riccati route.} Exactly as the $H_2$ control SDP~\eqref{eq:Jcont}
has an $O(r_x^3)$ Riccati twin ($J_{\mathrm{cont}}^\star=\operatorname{tr}(PW)$ from the
control DARE), the fixed-$\gamma$ SDP~\eqref{eq:Hinf_fixed} admits a Riccati twin that
avoids the interior-point solve entirely and hence useful at scale (but limits the kind of regularizers and constraints that can be added to the control and estimation design problems).

\section{The locally linearized case: MPC and the eKF}\label{sec:ekf}

The Bayesian filter is infinite dimensional; we approximate it by an extended Kalman
filter (eKF) \cite{anderson2012optimal,ramadan2024extended} propagating the first
two moments $x_{k\mid k},\Sigma_{k\mid k}$ (recursions in
Appendix~\ref{app:infostate}).

Concretely, with Jacobians $F_k = \partial f_\theta/\partial x|_{x_{k\mid k}}$ and
$H_k = \partial h_\theta/\partial x|_{x_{k\mid k-1}}$, the eKF alternates a
prediction and a measurement update steps similar to the Kalman filter.

If we stack the mean and the vectorized
upper-triangular elements of the state covariance into the information state vector
$\mathcal X_k = [x_{k\mid k}^\top,\, \operatorname{upvec}\Sigma_{k\mid k}^\top]^\top \in
\mathbb{R}^{n_\mathcal X}$, $n_\mathcal X = r_x + r_x(r_x+1)/2$, with lifted dynamics
induced by the (prediction-mode) eKF map (see Appendix~\ref{app:infostate} for more details)
\begin{equation} \label{eq:eKF as a dynamic system}
\mathcal{X}_{k+1}
     = \mathcal{F}_\theta(\mathcal{X}_k, u_k),
\end{equation}
and then linearize about the nominal information state (e.g. $(\mathcal{X}_0, u_0)$, since $k=0$ is current-time in receding-horizon convention)
\begin{equation*}
    \mathcal A_\theta = \left . \frac{\partial \mathcal F_\theta (\mathcal{X}, u_0)}{ \partial \mathcal{X}} \right |_{\mathcal{X} = \mathcal{X}_0},\quad \mathcal B_\theta = \left . \frac{\partial \mathcal F_\theta (\mathcal{X}_0, u)}{ \partial u} \right |_{u = u_0},
\end{equation*}
we end up in a linear system in the information state. The second-stage problem becomes $\widehat J_{sc} ^\star$, the approximate version of $J_{\mathrm{sc}}^\star$ in \eqref{eq:Jstoch_N},
\begin{equation}\label{eq:mpc}
\begin{aligned}
\widehat J^\star_{sc}(\theta) = &\min_{u_{0:N-1}}\
\frac{1}{N}\sum_{k=0}^{N-1} \mathcal L(\mathcal X_k,u_k) + \mathcal L^N(\mathcal X_N)\\
\text{s.t. } & \mathcal X_{k+1} = \mathcal A_\theta \mathcal X_k +
B_\theta u_k,\\
& \Sigma_{k\mid k} \succ 0,\quad \mathcal C(\mathcal X_{k+1},u_k) \le 0,
\end{aligned}
\end{equation}
where $\mathcal L$ now reconstructs the covariance term in $\operatorname{tr}(Q\Sigma_{k\mid k})$ from its vectorized upper-triangular values found in $\mathcal X_k$ (and similarily for the conic constraint $\Sigma_{k\mid k} \succ 0$). 
The function $\mathcal C(\cdot)$ represents the state/input constraints.

\emph{Design gradients from the dynamics costates.} Let $\lambda_k$ be the
multipliers (costates) of the dynamics equalities in \eqref{eq:mpc} in the corresponding Langrangian. By Theorem~\ref{prop:envelope}, the only explicit
$\theta$-dependence is through $\mathcal A_\theta, B_\theta$,
so
\begin{equation}\label{eq:gradmpc}
\frac{\partial \widehat J^\star_{sc}}{\partial \theta_i} = \sum_{k=0}^{N-1} \lambda_k^{\star\,\top} \Big(
\frac{\partial \mathcal A_\theta}{\partial \theta_i}\mathcal X_k^\star +
\frac{\partial \mathcal B_\theta}{\partial \theta_i} u_k^\star \Big).
\end{equation}

\begin{rem}
    If the state evolution and the measurement operator $f_\theta$ and $h_\theta$ represent a linear system, then $\mathcal{A}_\theta = \operatorname{blkdiag}\,(A_\theta, \bar A _\theta)$ (the evolution of the error covariance is independent of the state estimate) and $\mathcal{B}_\theta = [B_\theta^\top \, 0]^\top$ (the error covariance is independent of the input). This is a consequence of the separation principle between estimation and control for fixed $\theta$.
\end{rem}

\begin{rem}[The structure of the controller as a dual control]
    The MPC control \eqref{eq:mpc} is feedback in the information state $\mathcal{X}_k$, instead of $x_{k \mid k}$ only as in the SDP case. This is due to the fact that the separation principle does not hold in the nonlinear regime: the input and state estimate can have an effect on the evolution of the estimation error covariance. Hence, the resulting controller admits probing/caution effects, the characteristics of a dual control. It is possible to turn \eqref{eq:mpc} into a function of $x_{k \mid k}$ only by fixing the value of $\Sigma_{k \mid k}$ to, for example, $\lambda_{\mathrm{probing}} I$, turning $\lambda_{\mathrm{probing}} \geq 0$ into a probing (excitation for data collection) knob an operator can choose based on the circumstances.
\end{rem}

\begin{rem}[Where exactness stops: the nonlinear case]\label{rem:ekfbreaks}
Theorem~\ref{prop:envelope} applies verbatim to the SDPs, which are genuinely
convex. In the nonlinear regime it applies to the convex program \eqref{eq:mpc},
not to the underlying stochastic control value \eqref{eq:Jstoch_N}. Three gaps
separate the two: (i)~Filtering
approximation: the eKF is a finite-dimensional approximation to the exact Bayesian recursion
\eqref{eq:Tupdate} \eqref{eq:Mupdate} by its first two moments, exact only for
affine $f_\theta,h_\theta$ \cite{anderson2012optimal} (the Kalman filter case).
(ii)~the prediction-only mode of the eKF used along the MPC horizon, which drops the innovation feedback. (iii)~Further linearization: $\mathcal A_\theta,
\mathcal B_\theta$ are taken about the current-time ($k=0$) information trajectory, so
\eqref{eq:gradmpc}. Nevertheless, Section~\ref{sec:numerical} shows that in practice this approximate gradient can still be used effectively in design optimization.
\end{rem}

\begin{rem}[Sensor allocation]\label{rem:sensor}
Deciding which of the candidate sensors to deploy fits ContEst without
disturbing the validity of its gradients. (a)~Adding binary indicators $b\in\{0,1\}^{\#\text{sensors}}$ and
costs $c$ to the upper level, $\min_{\theta,b}
J_{\mathrm{sc}}^\star(\theta,b)+c^\top b$, leaves the inner problem, and hence the
envelope gradients \eqref{eq:chain}, \eqref{eq:gradmpc}, unchanged. The price is a
(NP-hard) mixed-integer program \cite{joshi2009sensor,tzoumas2016sensor}.
(b)~Relaxing $b$ to continuous gains $\alpha\ge0$ ($C_{\alpha,\theta}=\operatorname{diag}(\alpha)\,C_{\mathrm{base},\theta}$) with an $\ell_1$ penalty in the design cost,
\begin{equation}\label{eq:Jdes}
J_{\mathrm{des}}^\alpha(\theta) = J_{\mathrm{des}}(\theta) + \lambda_\alpha \| \alpha \|_1,
\quad \alpha \ge 0,
\end{equation}
promotes sparse sensing \cite{dhingra2014admm,saraf2017h2,das2017sparse}. (c)~The $\ell_1$ kink is removed
by its epigraph lift \cite{boyd2004convex}, $\lambda_\alpha \mathbf 1^\top t$ with
$\alpha\preceq t$, keeping $J$ differentiable so the BFGS guarantees of
Remark~\ref{rem:algorithm} hold. Finally, a threshold on $\alpha^\star$ recovers a binary
deployment.
\end{rem}

\begin{rem}[Outer optimization]\label{rem:algorithm}
Since the design gradient is available,
Problem~\ref{prob:two-stage} is solved by a first-order method, here a projected
BFGS iteration on $\theta$ over the box $\Theta$
\cite{nocedal2006numerical}. Because the gradients are
exact (not finite-difference), the
secant pairs driving the BFGS update are accurate. Under the usual smoothness
and local strong-convexity assumptions the outer iteration attains the local
superlinear rate of full BFGS \cite{nocedal2006numerical}. When the inner value
is only piecewise smooth (e.g.\ the $\ell_1$-selection or worst-case $H_\infty$
regimes) and the returned gradient is a subgradient, a nonsmooth BFGS variant
with a suitable line search still converges reliably in practice
\cite{lewis2013nonsmooth}.
\end{rem}

\section{On the computational complexity of ContEst with interior-point methods}
\label{sec:sim:complexity}

Because the envelope theorem (Section~\ref{sec:envelope}) yields the design
gradient from the primal and dual of a single inner solve, never by
differentiating through the solver's iterations, the per-design gradient is one inner
solve plus a matrix multiplication, and the outer BFGS iteration inherits the
complexity class of the inner solver.

\emph{Interior-point covariance/LMI solves scale as $O(r_x^6)$.} Both ContEst
inner problems place a covariance matrix among the decision variables. In the LQG
SDPs of Section~\ref{sec:lqg} the variable is the $r_x\times r_x$ covariance
$\Sigma$ (and its dual), i.e.\ $m=O(r_x^2)$ scalar unknowns; a primal-dual
interior-point method assembles and factorizes a dense Schur-complement (Newton)
system of size $m\times m$, at cost $O(m^3)=O(r_x^6)$ per iteration, times the
$O(\sqrt{r_x})$/$O(1)$ iterations interior-point methods take in practice. Empirically, solving the LQR
covariance program with the Clarabel interior-point solver took
$\approx\!2.6\,$s at $r_x=20$ and $\approx\!30\,$s at $r_x=30$: an
$(30/20)^6\!\approx\!11\times$ growth, matching the $O(r_x^6)$ prediction, so a
naive SDP inner problem caps ContEst at a few tens of states. To scale beyond that: (i) model structure and sparsity patterns can be exploited to reduce the number of nonzero entries in $\Sigma$, (ii) Krylov methods can be used within interior point implementations, instead of using direct Cholesky factorizations \cite{gondzio2012interior}.

\emph{The eKF--MPC inner problem meets the same wall.} There the decision trajectory
is the information state $\zeta=[x;\operatorname{upvec}\Sigma]$ of dimension
$r_x+\tfrac{r_x(r_x+1)}{2}=O(r_x^2)$; the horizon-$N$ MPC \eqref{eq:mpc} thus has
$O(N r_x^2)$ variables and its interior-point solve again grows steeply in $r_x$
(we measured $\approx\!33\,$s per evaluation at $r_x=32$).
Consequently the nonlinear/constrained studies of
Sections~\ref{sec:sim:distill}--\ref{sec:sim:pll} are run at modest $r_x$, and
large $r_x$ similarly requires model reduction, structure exploitation, and Krylov methods.

\section{Simulation studies}\label{sec:numerical}

We evaluate ContEst on four co-design problems spanning very different physics. All these examples were solved by the Clarabel
interior-point solver. The outer loop is the projected quasi-Newton method of
Remark~\ref{rem:algorithm} (BFGS) over the box $\Theta$, run from $5$ initial
points (the baseline design $\theta_{\mathrm{nom}}$ plus four uniform random
restarts in $\Theta$) keeping the best minimizer. The envelope-theorem gradients are exact, so we do not verify them
per study, but as a one-time sanity check we confirmed that
the analytic gradient matches central finite differences to a relative error below
$10^{-4}$.\footnote{ContEst obtains the exact design gradient at essentially the cost of a
single inner solve, while forward (central) differences need $r_\theta+1$ ($2r_\theta$)
inner solves for an $r_\theta$-vector $\theta$, for an approximate gradient.}

Code reproducing every study is publicly
available.\footnote{\url{https://github.com/msramada/ContEst}.}

For the studies in which linearization takes place, we report Monte-Carlo estimates of the realized closed-loop cost to confirm each optimized
design improves the true cost, not merely the surrogate problem. Starting from independent initial states sampled from $p_0$ and
with sampled process/measurement noise,
\begin{equation}\label{eq:costsplit}
\begin{aligned}
J^{MT}_{\mathrm{sc}} &= \frac{1}{M}\sum_{m=1}^M\sum_{k}^T\!\big(\|x_k\|_Q^2+\|u_k\|_R^2\big),\\
J^{MT}_{\mathrm{est}} &= \frac{1}{M}\sum_{m=1}^M\sum_{k}^T\|x_k-\hat x_{k\mid k}\|_Q^2,
\end{aligned}
\end{equation}
with $M$ the number of Monte-Carlo sample trajectories and $T$ each trajectory's
horizon length. These three studies use $M=100$ sample
trajectories over a horizon of $T=30$ steps, while the receding-horizon
MPC itself plans over $N=12$ future steps.

\subsection{Spacecraft attitude determination and control (ADCS)}
\label{sec:sim:adcs}

Pointing accuracy on every space mission is ultimately set by both the attitude estimator (the sensor suite: star trackers, gyros)
and the reaction-wheel authority. On a spacecraft these compete for a shared,
tightly constrained power/mass budget: a heavier, higher-power star tracker
or gyro reports lower noise, and a larger wheel delivers more torque, but every
watt and kilogram spent on one is unavailable to the others. We pose this on the
steady-state SDP path of Sections~\ref{sec:lqg} and \ref{sec:hinf}: $H_2$
state-feedback control under a hard cap on the steady-state control-effort
covariance, paired with worst-case ($H_\infty$) robust estimation, the
actuation and sensing axes tied by the shared budget.

A rigid satellite with inertia matrix
$J=\operatorname{diag}(4,6,5)\,\si{kg.m^2}$ regulates attitude and rate to zero
about its pointing equilibrium.
The state $x=[\phi,\theta_a,\psi,\omega_x,\omega_y,\omega_z]$ ($r_x=6$) is attitude and body rates; the input $u$ is three reaction-wheel
torques ($r_u=3$); the outputs are a star tracker (attitude) and a rate gyro
($r_y=6$, $C=I_6$). The only nonlinearity is the gyroscopic coupling $\omega\times
(J\omega)$, whose Jacobian vanishes at the pointing equilibrium $\omega=0$;
a single linearization there yields three decoupled double integrators.
With $\Delta t=\SI{0.1}{s}$,
\begin{subequations}\label{eq:adcs-model} (the $+$ superscript denotes the next time-step in a recursion)
\begin{align}
(\phi,\theta_a,\psi)^{+} &= (\phi,\theta_a,\psi) + \Delta t\,(\omega_x,\omega_y,\omega_z),\\
\omega^{+} &= \omega + \Delta t\,J^{-1}\big(e_{\mathrm{rw}}\,u - \omega\times(J\omega)\big),\\
h_\theta &= \big[\phi,\theta_a,\psi, \omega_x,\omega_y,\omega_z\big]^\top,
\end{align}
\end{subequations}
with $V_\theta=\operatorname{diag}(v_0/\alpha_{\mathrm{st}}^2\,I_3,\,v_0/\alpha_{\mathrm{gyro}}^2\,I_3), \, v_0=0.20$
(better precision $\Rightarrow$ lower variance), with
$W=\operatorname{diag}(10^{-6}I_3,10^{-4}I_3)$, and cost weights
$Q=\operatorname{diag}(4I_3,I_3)$ (attitude weighted four times more heavily than
rate, so the star tracker is intrinsically more valuable than the gyro) and
$R=0.1\,I_3$ (actuation is bounded by a hard effort cap, below, rather than by $R$).
The three design parameters are reaction-wheel authority $e_{\mathrm{rw}}$, which
scales the input and so enters $B_\theta$, and the star-tracker and rate-gyro precisions
$\alpha_{\mathrm{st}},\alpha_{\mathrm{gyro}}$, which enter the measurement
covariance $V_\theta$.

The control side is priced by the $H_2$ covariance SDP~\eqref{eq:Jcont} (dense,
$r_x=6$) in the stationary state covariance $\Sigma$, the input-second-moment
surrogate $Z_0$, and $L=K\Sigma$: its value
$J_{\mathrm{cont}}(\theta)=\operatorname{tr}(Q\Sigma)+\operatorname{tr}(RZ_0)$
equals the LQG cost, with the design gradient in $e_{\mathrm{rw}}$ read as an exact
gradient from the LMI dual~\eqref{eq:gradAB} (no solver differentiation). We add one
\emph{hard} constraint, $\operatorname{tr}(Z_0)\le u_{\max}$: since
$Z_0\succeq K\Sigma K^\top$ is exactly the stationary control-effort covariance (the
steady-state input second moment), this caps the steady-state control-effort
covariance directly (a linear LMI that an LQR/Riccati weighting cannot
express, because $R$ can only price effort, not bound it). We set
$u_{\max}=0.6\,\operatorname{tr}(Z_0^{\mathrm{unc}})$ (60\% of the unconstrained
baseline effort). The sensing side is priced by the fixed-$\gamma$
$H_\infty$ robust-filter SDP at attenuation
$\gamma^2=4$, chosen as five times the minimal feasible $\gamma^2_{\min}=0.8$ so the
BRL LMI is feasible and its dual nondegenerate
(Corollary~\ref{cor:hinf_fixed}). Because $e_{\mathrm{rw}}$ enters only $B_\theta$ and the precisions only $V_\theta$, the two inner
values are coupled solely through the shared budget in the design cost (a
design-level) rather than dynamics-level coupling.

The design cost is a shared budget $c_b(e_{\mathrm{rw}}+\alpha_{\mathrm{st}}+
\alpha_{\mathrm{gyro}}-B)^2$ ($c_b=5$, $B=3$), with box
$e_{\mathrm{rw}}\in[0.3,3]$, $\alpha_{\mathrm{st}},\alpha_{\mathrm{gyro}}\in[0.3,5]$,
and baseline $\theta_{\mathrm{nom}}=\mathbf 1$ (budget spent uniformly). On the SDP
path actuation is bounded by the hard control-effort-covariance cap
$\operatorname{tr}(Z_0)\le u_{\max}$ rather than by pointwise input constraints on
the trajectory.

All five starts converged to the same $\theta$.
The co-design (Table~\ref{tab:adcs}) reallocates the budget:
it upgrades the star tracker most ($\alpha_{\mathrm{st}}\!:\,1\to1.70$), holds wheel
authority essentially fixed ($e_{\mathrm{rw}}\!:\,1\to1.005$ (the hard effort cap
saturates the return on actuation, so extra authority buys nothing), and drives the lightly weighted
rate gyro to its floor ($\alpha_{\mathrm{gyro}}\!:\,1\to0.30$), spending that budget
where it pays most. Improvements in control, estimation and design costs are reported in Table~\ref{tab:formulations}.

\begin{table}[t]
\centering
\caption{Spacecraft ADCS ($H_2$ control under a control-effort-covariance cap $+$
$H_\infty$ robust-estimation SDP, single
linearization about the pointing equilibrium): baseline (uniform budget) vs.\
co-designed allocation of the shared power/mass budget
$e_{\mathrm{rw}}+\alpha_{\mathrm{st}}+\alpha_{\mathrm{gyro}}\approx B=3$.}
\label{tab:adcs}
\begin{tabular}{llcc}
\hline
parameter & axis & baseline & optimal \\
\hline
$e_{\mathrm{rw}}$       & $f$ (reaction wheels) & 1.000 & 1.005 \\
$\alpha_{\mathrm{st}}$  & $V$ (star tracker)    & 1.000 & 1.703 \\
$\alpha_{\mathrm{gyro}}$& $V$ (rate gyro)       & 1.000 & 0.300 \\
\hline
\end{tabular}
\end{table}

\subsection{Binary distillation column}
\label{sec:sim:distill}

Distillation is the workhorse separation of the process
industries and its control is dominated by two classical layout decisions: on
which tray to introduce the feed, and on which tray to place the
inferential (temperature) sensor that is used to recover the composition
profile. These locations shape $f_\theta$ and $h_\theta$, so this is a ContEst problem.

A twelve-stage column (stage~1 reboiler
until stage~12 condenser) separates a light/heavy mixture; the state
$x=[\Delta x_1,\dots,\Delta x_{12}]$ ($r_x=12$) collects stage light-component
mole-fraction deviations about the nominal profile, in deviation form so
$f_\theta(0,0)=0$ for every $\theta$. The nonlinearity enters the dynamics twice: the vapour-liquid equilibrium
$y(x)=\alpha x/(1+(\alpha-1)x)$ ($\alpha=2.5$) and the bubble-point temperature
$T(x)=T_{B0}-\Delta T_B\,y(x)$. The single input is the reboiler boil-up
deviation $u=\Delta \bar V$ ($r_u=1$), limited to $|\Delta \bar V|\le1.5$; the measurements
($r_y=3$) are three Gaussian-weighted stage temperatures read at the sensor
locations. Integration is at $\SI{0.05}{s}$; the receding-horizon MPC plans over
$N=12$ steps. The design carries
$N_f=3$ feed-stage locations $\theta_f=(f_1,f_2,f_3)$ and $N_s=3$
temperature-sensor locations $\theta_h=(s_1,s_2,s_3)$, all in the interval $[2,11]$. Placement is nearly free to relocate,
so $J_{\mathrm{des}}$ carries only a mild regularisation
$\lambda_{\mathrm{des}}\lVert\theta-\theta_{\mathrm{nom}}\rVert^2$ about the
default $\theta_{\mathrm{nom}}$ with feeds at $(3,4.5,6)$ and sensors at
$(2,3.5,5)$.

With equilibrium
deviations $\Delta y_i=y(x_i^{\mathrm{ss}}+\Delta x_i)-y_i^{\mathrm{ss}}$, for
$i=1,\dots,12$ and $j=1,\dots,3$,
\begin{subequations}\label{eq:distill-model}
\begin{align}
\Delta x_i^{+} &= \Delta x_i + \tfrac{\Delta t}{\bar M}\big(\mathrm{liq}_i+\mathrm{vap}_i\big),\\
\mathrm{liq}_i &= L_i(\theta_f)\,(\Delta x_{i+1}-\Delta x_i),\\
\mathrm{vap}_i &= (\bar V+u)\,(\Delta y_{i-1}-\Delta y_i) + u\,(y_{i-1}^{\mathrm{ss}}-y_i^{\mathrm{ss}}),\\
h_{\theta,j} &= \textstyle\sum_{i=1}^{12}\hat q_i(s_j)\,\big(T(x_i^{\mathrm{ss}}+\Delta x_i)-T_i^{\mathrm{ss}}\big),
\end{align}
\end{subequations}
where the superscript $(\cdot)^{\mathrm{ss}}$ denotes the nominal steady-state
operating point about which the model is written in deviation form. The
boundary terms $\Delta x_{13}=\Delta y_0=y_0^{\mathrm{ss}}=0$, feed weights
$q_l(\theta_f)\propto \sum_{j=1}^{3}e^{-(l-f_j)^2/2\sigma_f^2}$ (normalized)
setting the internal liquid traffic
$L_i(\theta_f)=L_{\mathrm{ref}}+F\sum_{l>i}q_l$, and sensor weights
$\hat q_i(s_j)\propto e^{-(i-s_j)^2/2\sigma_{\mathrm{sens}}^2}$. Constants $\alpha=2.5$,
$L_{\mathrm{ref}}=2$, $\bar V=2.5$, $F=2.5$, $\bar M=1$,
$\sigma_f=\sigma_{\mathrm{sens}}=1.2$, $T_{B0}=\SI{380}{K}$,
$\Delta T_B=\SI{40}{K}$. The noise covariances $W=2\!\times\!10^{-5}I_{12}$,
$V=0.6\,I_3$ ($\si{K^2}$), and prior $\Sigma_0=10^{-2}I_{12}$. The cost weights
$Q_{ii}=20+30\,e^{-(i-6.5)^2/2\sigma_Q^2}$ ($\sigma_Q=2.4$, heaviest
mid-column), $R=2$,
$|u|_\infty\le1.5$, with a mild placement regularization
$\lambda_{\mathrm{des}}\lVert\theta-\theta_{\mathrm{nom}}\rVert^2$.

The design landscape is multimodal: the five-start BFGS converge to three distinct
minima, and the baseline start alone converges to an inferior one, so the random
restarts are what recover the reported optimum. ContEst spreads the feed
streams: two settle lower toward the reboiler while
the third climbs high up the column, and raises the three
temperature sensors out of the bottom region across the mid-and-upper trays, covering the composition profile rather than clustering low (Table~\ref{tab:distill}). Improvements on the control, estimation and design costs are shown in Table~\ref{tab:formulations}.

\begin{table}[t]
\centering
\caption{Binary distillation: naive baseline vs.\ co-designed placements. All six
parameters are continuous tray locations (not gains); baseline $\theta_{\mathrm{nom}}$
places feeds at $(3,4.5,6)$ and sensors at $(2,3.5,5)$.}
\label{tab:distill}
\begin{tabular}{llcc}
\hline
parameter & axis & baseline & optimal \\
\hline
$f_1$ & $f$ (feed stage)    & 3.000 & 2.479 \\
$f_2$ & $f$ (feed stage)    & 4.500 & 3.584 \\
$f_3$ & $f$ (feed stage)    & 6.000 & 9.018 \\
$s_1$ & $h$ (temp.\ sensor) & 2.000 & 2.972 \\
$s_2$ & $h$ (temp.\ sensor) & 3.500 & 5.428 \\
$s_3$ & $h$ (temp.\ sensor) & 5.000 & 8.010 \\
\hline
\end{tabular}
\end{table}

\subsection{PLL/DSE tuning in a low-inertia grid}
\label{sec:sim:pll}

In this example we run ContEst as an online adaptive tool to co-design (tune) the state estimator itself: the per-inverter
dynamic state estimation (DSE) and the damping authority. Every
grid-following (GFL) inverter estimates grid frequency and phase with a
phase-locked loop (PLL), whose bandwidth is a classic bias-variance knob: a
narrow-band PLL is smooth but sluggish and lags real frequency excursions, while
a wide-band PLL tracks fast but passes more measurement noise into the frequency
estimate. In low-inertia, inverter-dominated grids, this tuning is decisive
for stability and for fast frequency response \cite{zhao2019power}. The PLL
bandwidth is therefore an estimator design variable that classical control
co-design cannot represent and sensor placement cannot capture.

Three GFL-inverter buses form a low-inertia chain. The state
$x=[\Delta\delta_{1:3},\Delta\omega_{1:3},\Delta\hat\omega^{\mathrm p}_{1:3}]$
($r_x=9$) stacks bus angle deviations, frequency deviations, and the per-bus PLL frequency estimate $\Delta\hat\omega^{\mathrm p}_i$; the input
$u=[\Delta P_{1:3}]$ is the inverter damping power ($r_u=3$). Each bus is measured
by a phasor measurement unit (PMU) angle channel and its PLL frequency channel
($r_y=6$). The genuine
nonlinearity is the synchronizing power $\propto\sin(\delta^{\mathrm{eq}}_{ij}+
\Delta\delta_i-\Delta\delta_j)$, in deviation form so $f_\theta(0,0)=0$ for
all $\theta$; integration is at $\SI{0.05}{s}$, the MPC plans over $N=12$ steps ahead. The design parameters
are the per-bus inverter damping effectiveness
$\theta_f=(e_1,e_2,e_3)\in[0.3,3]^3$ and the per-bus PLL bandwidths
$\theta_h=(b_1,b_2,b_3)\in[2,20]\,\si{rad/s}$.

Since the PLL bandwidths and damping gains are software values, and the useful allocation is operating-point dependent (the localized disturbance
fixes which bus is fastest and least certain), the low per-solve cost of
ContEst makes it natural to be used as an event-triggered adaptation layer that automatically tunes system parameters after sudden changes in the environment, such as the following scenario.

First, under normal,
undisturbed operation (equilibrium, no post-fault kick, and a uniform
quiet-level prior $\Sigma_0$ at every bus), ContEst is run once from the naive
default $\theta_{\mathrm{nom}}$ ($e_i=1$, uniform $b_i=6$), producing a
peacetime design $\theta_{\mathrm{peace}}$. Then a sudden event
concentrates a disturbance at bus~1 (large post-fault kick and an elevated
prior $\Sigma_0$ there), while buses~2,3 stay `quiet'; this second trigger
re-solves the same co-design problem, now started from
$\theta_{\mathrm{peace}}$ rather than from $\theta_{\mathrm{nom}}$, producing
the post-fault design $\theta^\star$. The automatic design/tuning
must trade tracking speed against readout noise per bus in both triggers.

For $i=1,2,3$,
\begin{subequations}\label{eq:pll-model}
\begin{align}
&\Delta\delta_i^{+} = \Delta\delta_i + \Delta t\,\Delta\omega_i,\\
&\Delta\omega_i^{+} = \Delta\omega_i + \tfrac{\Delta t}{2H_i}
   \big(\!-P_i^{\mathrm e}(\Delta\delta) - D_i\Delta\omega_i + e_i u_i\big),\\
&\Delta\hat\omega^{\mathrm p,+}_i = \Delta\hat\omega^{\mathrm p}_i
   + \Delta t\,b_i\,(\Delta\omega_i - \Delta\hat\omega^{\mathrm p}_i),\\
&h_\theta = \big[\,\Delta\delta_j;\ \Delta\hat\omega^{\mathrm p}_j\,\big]_{j=1}^{3},
\end{align}
\end{subequations}
where $P_i^{\mathrm e}(\Delta\delta)=\sum_{j\neq i}K_{ij}\big[\sin(\Delta^{\mathrm{eq}}_{ij}
+\Delta\delta_i-\Delta\delta_j)-\sin\Delta^{\mathrm{eq}}_{ij}\big]$, with
$\Delta^{\mathrm{eq}}_{ij}=\delta^{\mathrm{eq}}_i-\delta^{\mathrm{eq}}_j$ the
equilibrium angle difference, is the synchronizing
power. Here $H=(1.6,2.2,2.0)\,\si{s}$ are the (low) inertia constants,
$D=(0.4,0.6,0.5)$ the natural dampings, $\delta^{\mathrm{eq}}=(0.15,-0.05,0.10)\,
\si{rad}$ the equilibrium angles, and the line synchronizing coefficients are
$K_{12}=1.2$, $K_{23}=1.1$, $K_{13}=0.4$ (an open $1$--$2$--$3$ chain with a weak
$1$--$3$ tie). The third line is the first-order PLL tracking the true
frequency with bandwidth $b_i$. The process noise covariance $W=\operatorname{blkdiag}
(2\!\times\!10^{-5}I_3,5\!\times\!10^{-4}I_3,5\!\times\!10^{-4}I_3)$ and the prior
$\Sigma_0$ (diagonal) large entry (1,1) at bus~1 ($3\!\times\!10^{-2}$), small elsewhere
($6\!\times\!10^{-3}$). The cost weights $Q=\operatorname{blkdiag}(15I_3,120I_3,0_3)$
(weighting the true angle and, heavily, the true frequency; PLL states
unweighted), $R=I_3$, $|u|_\infty\le1.5$. The design cost 
$J_{\mathrm{des}}=\lambda_e\sum_i(e_i-1)^2+\lambda_b\sum_i(b_i-b_{\mathrm{nom}})^2$,
$\lambda_e=0.5$, $\lambda_b=0.02$, $b_{\mathrm{nom}}=6$. The design coupling is
that the PLL bandwidth $b_i$ enters both $f_\theta$ (the tracking line above)
and the measurement covariance, which is design-dependent,
\begin{equation}\label{eq:pll-V}
V_\theta=\operatorname{blkdiag}\!\big(\sigma_\delta^2 I_3,\
\operatorname{diag}_i\,\sigma_\omega^2(1+\kappa\,b_i)\big),
\end{equation}
$\sigma_\delta^2=4\!\times\!10^{-3}$, $\sigma_\omega^2=2\!\times\!10^{-3}$,
$\kappa=0.30$: a wider-band PLL tracks faster but yields a noisier frequency
readout---the bias--variance knob $\theta_h$ resolves.

Under normal operation all five starts converge
to the same minimizer $\theta_{\mathrm{peace}}$. With little disturbance for
the actuator to counteract, the inverter damping is left exactly at its
baseline ($e_i=1$), while every PLL
bandwidth is raised uniformly to trim ambient estimation error
($b_i\!:\,6\to10.6$). This cuts $J_{\mathrm{est}}$ by $12.1\%$
($38.0\to33.4$) and $J_{\mathrm{cont}}$ by only $4.2\%$ ($83.6\to80.1$), for a
$5.6\%$ reduction in the total $J_{\mathrm{des}}$ ($121.5\to114.7$).

\emph{Trigger (post-fault).} The bus-1 fault re-triggers ContEst. All five starts, including one seeded at $\theta_{\mathrm{peace}}$ itself, converged to the same $\theta^\star$. ContEst reallocates sharply: it
raises the disturbed bus's PLL bandwidth further ($b_1\!:\,10.6\to13.7$) while
leaving the quiet buses' bandwidths essentially where trigger 1 left them
($b_{2,3}\!:\,10.6\to10.6$), and (now that there is a real disturbance to
counteract) unlocks the inverter damping authority for the first time
($e_1\!:\,1\to3.0$, $e_{2,3}\!:\,1\to1.95,2.19$; none of the bandwidths reach
the upper bound of $20$. Results are summarized in Table~\ref{tab:formulations}.

\begin{table}[t]
\centering
\caption{Low-inertia grid (PLL/DSE): naive baseline $\theta_{\mathrm{nom}}$,
trigger-1 peacetime design $\theta_{\mathrm{peace}}$ (ContEst run under normal
operation), and trigger-2 post-fault design $\theta^\star$ (ContEst re-run,
started from $\theta_{\mathrm{peace}}$, once the bus-1 fault hits). The PLL
bandwidths $b_i$ enter both $f_\theta$ and the measurement covariance
$V(\theta)$.}
\label{tab:pll}
\begin{tabular}{llccc}
\hline
parameter & axis & $\theta_{\mathrm{nom}}$ & $\theta_{\mathrm{peace}}$ & $\theta^\star$ \\
\hline
$e_1$ & $f$ (inverter damping) & 1.000 & 1.000 & 3.000 \\
$e_2$ & $f$ (inverter damping) & 1.000 & 1.000 & 1.951 \\
$e_3$ & $f$ (inverter damping) & 1.000 & 1.000 & 2.191 \\
$b_1$ & $h,V$ (PLL bandwidth)  & 6.000 & 10.578 & 13.718 \\
$b_2$ & $h,V$ (PLL bandwidth)  & 6.000 & 10.569 & 10.567 \\
$b_3$ & $h,V$ (PLL bandwidth)  & 6.000 & 10.575 & 10.575 \\
\hline
\end{tabular}
\end{table}

\subsection{Multi-terminal HVDC: sparse sensing}
\label{sec:sim:mtdc}

Multi-terminal HVDC (MTDC) grids regulate DC-bus voltage with droop
control: each converter reacts to its own local voltage deviation by
modulating its power injection, a fixed-structure, decentralized, proportional
law (the DC-grid analogue of frequency droop) rather than a freshly
synthesized feedback gain. Droop needs no communication, but an aggressive
droop setting on a meshed grid can excite a lightly-damped
inter-terminal resonance, a worst-case (not average-case) failure mode
that $H_\infty$ prices directly. On the sensing side, not every terminal needs
a telemetered voltage sensor: the ring couples neighbouring buses, so an
unsensed terminal's converter current can be reconstructed from its
neighbours' voltages, trading estimation quality against a hard per-sensor
telemetry cost. This is the same $H_2$/$H_\infty$ SDP machinery of
Sections~\ref{sec:lqg} and \ref{sec:hinf}, at network scale ($r_x=30$) and under
an explicit sparsity pattern that reflects the physical topology, exercising
a structure-exploitation strategy to reduce the computational complexity of each solve.

A ring of $15$ terminals has state $x=[\Delta V_i,\Delta I_i]_{i=1}^{15}$
($r_x=30$, bus-voltage and converter-current deviations), input $u$ the
per-terminal supplementary power-reference modulation ($r_u=15$), and output
$y$ the (possibly gain-scaled or absent) telemetered bus voltages
($r_y\le15$). With $\Delta t=\SI{5}{ms}$ and terminal $i$'s ring neighbours
$(i^-,i^+)$,
\begin{subequations}\label{eq:mtdc-model}
\begin{align}
\Delta V_i^+ &= \Delta V_i + \tfrac{\Delta t}{C_i}\big[G_{\mathrm{dc}}(\Delta V_{i^-}\!+\!\Delta V_{i^+}\!-\!2\Delta V_i) + \Delta I_i\big] + w_i,\\
\Delta I_i^+ &= \Delta I_i + \Delta t\,\big[u_i - \Delta I_i - k_0\theta_i\,\Delta V_i\big]/\tau_i,\\
h_{\theta,i} &= \alpha_i\,\Delta V_i,
\end{align}
\end{subequations}
where $\theta_i=k_i$ is the droop gain at terminal $i$ (entering the
$f_\theta$ axis) and $\alpha_i\in[0,1]$ is a per-terminal sensor gain
(entering $C$, the $h_\theta$ axis; $\alpha_i=0$ removes the sensor). Line
conductance $G_{\mathrm{dc}}=5$, droop leverage $k_0=12$, and heterogeneous
per-terminal capacitances $C_i$, converter lags $\tau_i$, and infeed
intensities $\sigma_i$ (driving the process-noise variance on $\Delta V_i$) are
fixed physical constants. Cost weights penalize voltage deviation far above
current ($Q=\operatorname{diag}(60,0.2,\dots)$, alternating per terminal),
$R=2I_{15}$, and a small per-sensor measurement variance $v=10^{-3}$.

\emph{Control: block-diagonal $H_\infty$.} $\theta$ enters $A_\theta$, so its
worst-case cost is priced by the fixed-$\gamma$ BRL SDP~\eqref{eq:Hinf_fixed}
at $\gamma^2=20$. A dense $(Y,L)$ solve is the $O(r_x^6)$ wall of
Section~\ref{sec:sim:complexity} in practice (1380 variables,
$\approx\!226\,$s); restricting $Y,L$ to a block-diagonal
($2\times2$ per terminal) pattern cuts this to $\approx\!0.4\,$s, results in a $+54.9\%$ conservatism cost in the reported
worst-case value but a clean, strictly complementary dual.

\emph{Sensing: banded $H_2$ with $\ell_1$ sparsity penalty.} Here, the estimation
SDP~\eqref{eq:Jest} is restricted to the banded
(ring-neighbour) pattern that does mirror the physical topology (instead of the block-diagonal case as in the control case above), it matches
the full dense SDP's true cost $\operatorname{tr}(Q\Sigma^\varepsilon)$ to
within $+0.06\%$ at a fraction of the variables, and its dual gives the exact
gradient in both $\theta$ (through $A_\theta$) and the sensor gains $\alpha$ (through
$C$). An $\ell_1$ penalty $\lambda_s\lVert\alpha\rVert_1$ ($\lambda_s=0.05$)
in the design cost promotes sparse sensor placement, in the smooth epigraph
form of Section~\ref{sec:sim:adcs}'s reasoning (here $\alpha\ge0$ already, so
the penalty is simply linear).

Because $\theta$ enters both inner values through $A_\theta$, the two axes are
dynamically coupled; because $\alpha$ enters only $C$, the joint search
alternates a $\theta$-step (fixed $\alpha$, both SDPs) and an $\alpha$-step
(fixed $\theta$, estimation SDP only), each well-conditioned, rather than one
monolithic 30-variable solve. The design cost also prices droop stress,
$c_k\sum_i\theta_i$ ($c_k=5\times10^{-3}$), with box
$\theta_i\in[0.5,4]$, $\theta_{\mathrm{nom}}=3$ (uniform droop, all 15 sensors
on).

All five starts converged to a design within $1.5\%$ of each other; the best
retains 12 of 15 sensors (dropping terminals 2, 6, 8) and spreads droop over
$\theta^\star\in[0.89,1.67]$, well below the uniform baseline. This cuts the
estimation cost $J_{\mathrm{est}}$ by $31.8\%$ ($0.650\to0.443$) and the
control cost $J_{\mathrm{cont}}$ by $33.4\%$ ($2.476\to1.650$), for a $33.0\%$
reduction in $J_{sc}=J_{\mathrm{cont}}+J_{\mathrm{est}}$ ($3.126\to2.093$) and,
after the droop-stress and sensor-sparsity price, a $21.6\%$ reduction in the
total $J_{\mathrm{des}}$ ($3.351\to2.626$). As in the other studies, these are
the exact steady-state SDP/Lyapunov values, not Monte-Carlo estimates: unlike
ADCS's single linearization or the eKF-MPC studies' receding-horizon surrogate,
this plant is linear by construction with no unmodeled dynamics, so there is
no surrogate-vs-truth gap for a Monte-Carlo check to reveal.

\begin{table}[t]
\centering
\caption{MTDC sensor-count frontier: fixed-droop (baseline $\theta_{\mathrm{nom}}=3$)
greedy sensor pruning vs.\ co-designing the droop at each retained-sensor
count (same sensor sets both ways, so the gap isolates the co-design gain).}
\label{tab:mtdc-frontier}
\small
\begin{tabular}{lcccc}
\hline
& \multicolumn{2}{c}{$J_{\mathrm{est}}$} & \multicolumn{2}{c}{$J_{\mathrm{cont}}$} \\
\#sensors & pruned & co-des. & pruned & co-des. \\
\hline
15 & 0.650 & 0.392 & 2.476 & 1.596 \\
9  & 0.770 & 0.460 & 2.476 & 1.645 \\
6  & 0.927 & 0.560 & 2.476 & 1.645 \\
4  & 1.082 & 0.653 & 2.476 & 1.644 \\
3  & 1.170 & 0.697 & 2.476 & 1.645 \\
\hline
\end{tabular}
\end{table}

Table~\ref{tab:mtdc-frontier} is the headline result: co-designing the droop
and sensor budget beats naive pruning at every count. The co-designed droop with only \textbf{4}
sensors ($J_{\mathrm{est}}=0.653$) already is close to the fixed-droop baseline's
\textbf{15}-sensor value ($J_{\mathrm{est}}=0.650$), i.e.\ $73\%$ of the
voltage telemetry is removed with almost no loss in estimation quality, purely by
retuning the actuation-side droop. A sensor-selection method with no actuation
axis \cite{joshi2009sensor,tzoumas2016sensor} cannot pose this trade: droop
reshapes the dynamics the estimator has to track, so a better-tuned
droop makes the remaining sensors' information go further.

Table~\ref{tab:formulations} consolidates the four studies (dimensions, the design
split $\theta=(\theta_f,\theta_h)$, the inner value, and
the reductions); the absolute baseline$\to$optimum costs are reported compactly with
each study above.

\begin{rem}
A caveat on interpretation: the baseline parameter values $\theta_{\mathrm{nom}}$
above were chosen to be plausible rather than carefully hand-tuned, so the
reported reductions should not be read as ContEst outperforming a state-of-the-art
design. The contribution is not a competing tuning heuristic but the exact,
dual-based design gradient that lets a first-order method search for the
joint optimum.  
\end{rem}

\begin{table*}[t]
\centering
\caption{Consolidated view of the four studies}
\label{tab:formulations}
\begin{tabular}{lcclcccc}
\hline
study & $(r_x,r_u,r_y)$ & $\theta$: $n_f{+}n_h$ (enters) & inner value
& $J_{\mathrm{est}}\!\downarrow$ & $J_{\mathrm{sc}}\!\downarrow$ & $J_{\mathrm{tot}}\!\downarrow$\\
\hline
ADCS & $(6,3,6)$   & $1{+}2$: $f_\theta$, $V_\theta$ & $H_2^{\mathrm{u\,capped}} + H_\infty$ & $42.7\%$ & $18.3\%$  & $\mathbf{34.0\%}$\\
Distillation      & $(12,1,3)$  & $3{+}3$: $f_\theta$, $h_\theta$ & eKF-MPC & $35.3\%$ & $5.6\%$  & $\mathbf{9.5\%}$\\
PLL/DSE    & $(9,3,6)$   & $3{+}3$: $f_\theta$, $(f_\theta,V_\theta)$ & eKF-MPC & $17.6\%$ & $39.5\%$ & $\mathbf{33.9\%}$\\
MTDC & $(30,15,\le\!15)$ & $15{+}15$: $f_\theta$, $h_\theta$ & $H_\infty^{\mathrm{blockdiag}} + H_2^{\mathrm{banded}}$ & $31.8\%$ & $33.4\%$ & $\mathbf{21.6\%}$\\
\hline
\end{tabular}
\end{table*}

\section{Conclusion}\label{sec:conclusion}

We introduced ContEst, a co-design framework that brings state and parameter
estimation into control co-design by optimizing system and sensing parameters against a stochastic (dual-control) objective
expressed by the information state. The inner value takes a convex form in
two regimes: a pair of semidefinite programs in the average-case LQG ($H_2$)
setting or their worst-case $H_\infty$ analog (a bounded-real-lemma semidefinite program). Sensor selection enters as a
mixed-integer problem with a sparse $\ell_1$ relaxation whose duals price the
essential sensors. In every case the design gradient is obtained exactly and cheaply
from the inner dual variables, the LMI duals, by the envelope
theorem, sidestepping both implicit differentiation of the solution map and
differentiation of the KKT system. 

Our nonlinear extension relies on the eKF as the approximate Bayesian filter, applies further linearizations of the information state dynamics to achieve a convex MPC in these dynamics, and the design 'gradient' is then constructed from the returned dynamics costates. While this 'gradient' proved helpful in our numerical examples, the multiple layers of approximations and linearizations used in its construction mean that it is not an exact gradient.

Future work includes Bayesian filters beyond the eKF and the handling of the nonlinear problem with less approximations, or quantifying some error bounds on the returned gradients achieved under approximation. We also target including chance-constrained and scenario-based inner problems, and the decomposition of the outer problem along multidisciplinary-design-optimization architectures for large interconnected
systems.


\appendix

\section{The information state and its LQG and eKF specializations}\label{app:infostate}

The second stage is a stochastic (dual) control problem
\cite{feldbaum1960dual,bar1974dual}. A causal law must be a function of the
available information $\mathcal{Z}_k$, not of the unobserved $x_k$; the
sufficient statistic is the filtering density $p(x_k \mid \mathcal{Z}_k;\theta)$,
governed by the Bayesian filter, which alternates a time update,
\begin{equation}\label{eq:Tupdate}
p(x_{k+1} \mid \mathcal{Z}_k, u_k) = \int p(x_{k+1} \mid x_k, u_k;\theta)\,
p(x_k \mid \mathcal{Z}_k)\, dx_k,
\end{equation}
and a measurement update,
\begin{equation}\label{eq:Mupdate}
p(x_{k+1} \mid \mathcal{Z}_{k+1}) \propto p(y_{k+1} \mid x_{k+1};\theta)\,
p(x_{k+1} \mid \mathcal{Z}_k, u_k).
\end{equation}
Both \eqref{eq:Tupdate}--\eqref{eq:Mupdate} involve the output
equation~\eqref{eq:ss-output} and hence the sensing parameters in $\theta$. The
standard smoothing identity makes this explicit in the cost.

\begin{lem}[Smoothing {\cite[Ch.~10]{resnick2019probability}}]\label{lem:smoothing}
For a probability space $(\Omega,\mathcal{A},\mathbf{P})$, a measurable
$\chi:\Omega\to\mathbb{R}$ with $\mathbb{E}_{\mathbf P}|\chi|<\infty$, and
$\sigma$-algebras $\mathcal{A}_0 \subset \mathcal{A}_1 \subset \mathcal{A}$,
$\mathbb{E}_{\mathbf P}\{\chi \mid \mathcal{A}_1\} =
\mathbb{E}_{\mathbf P}\{\mathbb{E}_{\mathbf P}\{\chi \mid \mathcal{A}_0\} \mid
\mathcal{A}_1\}$.
\end{lem}

\begin{lem}\label{lem:stagecost}
For a causal law $u_k = u_k(\mathcal{Z}_k, p_0)$,
\begin{equation*}
\mathbb{E}\,\ell(x_k,u_k) = \!\iint \ell(x_k,u_k)\,
p(x_k \mid \mathcal{Z}_k;\theta)\, p(\mathcal{Z}_k;\theta)\, dx_k\, d\mathcal{Z}_k.
\end{equation*}
\end{lem}
\begin{proof}
Conditioning on $\mathcal{Z}_k$ yields a $\sigma$-algebra contained in the one
underlying \eqref{eq:Jstoch_N}; apply Lemma~\ref{lem:smoothing}.
\end{proof}

Therefore, the term $\E \| x_k \| _Q^2$ can be expressed by
\begin{align*}
    \E \| x_k \| _Q^2 &= \E \left \{ \E \left \{  \| x_{k}\|_Q^2 \mid \mathcal{Z}_k \right \} \right \}\\
    &=\E \left \{  \operatorname{tr}(Q \Sigma_{k \mid k}) +  \| x_{k \mid k}\|_Q^2 \right \},
\end{align*}
where the first two conditional moments are
\begin{equation}
\begin{aligned}\label{eq:filterMeanCov}
     x_{k \mid k} &= \E \left \{ x_k \mid \mathcal Z_{k}\right \},\\
    \Sigma_{k \mid k} &= \E \left \{ [x_k-x_{k \mid k}][x_k-x_{k \mid k}]^\top \mid \mathcal Z_{k} \right \},
\end{aligned}
\end{equation}
are substituted above after writing $x_k = x_{k \mid k} + \varepsilon_k$, where $\varepsilon_k$ is the estimation error (of zero mean), and using the cyclic property of the trace operator.

Hence, the cost functions \eqref{eq:Jstoch_N} and \eqref{eq:Jstoch_infty} can be represented by 
\begin{align*}
     &J_{\mathrm{sc}} = \frac{1}{N} \cdot \E \Bigg [ \| x_{N\mid N} \| ^2_ Q  + \operatorname{tr}(Q \Sigma_{N \mid N})  \nonumber\\
    &+ \sum_{k=0}^{N-1} \left [  \| x_{k\mid k} \|^2_ Q  + \|u_k \|_R^2 + \operatorname{tr}(Q \Sigma_{k \mid k}) \right ] \Bigg ],
\end{align*}
or in the infinite-horizon case
\begin{align*}
     &J_{\mathrm{sc}} = \lim_{k \to \infty} \E  \left [  \| x_{k\mid k} \|^2_ Q  + \|\kappa(\mathcal{Z}_k) \|_R^2 + \operatorname{tr}(Q \Sigma_{k \mid k}) \right ].
\end{align*}

In this paper we have two regimes, linear and nonlinear. For the linear cases (Sections~\ref{sec:lqg} and \ref{sec:hinf}), we adopt the infinite horizon cost, with $u_k = \kappa(\mathcal{Z}_k) = K x_{k \mid k}$, that is, certainty equivalent state-feedback control (due to the separation principle: the state error covariance $\Sigma_{k \mid k}$ evolution is not dependent on $u_k$. After substituting the feedback control, and using the cyclic property of the trace again, the cost becomes that in \eqref{eq:linear stoch cost}, given that $\Sigma^\epsilon = \lim_{k \to \infty} \E \varepsilon_k \varepsilon_k^\top = \lim_{k \to \infty} \Sigma_{k \mid k}$, and $\Sigma = \lim _{k \to \infty} \E x_{k \mid k} x_{k \mid k}^\top$, the estimation error covariance and the state covariance (about the origin), respectively.

For the nonlinear regime, we choose the eKF to track the evolution of the state first two moments. In addition to \eqref{eq:filterMeanCov}, let
\begin{align*}
     x_{k \mid k-1} &= \E \left \{ x_k \mid Z_{k-1}, u_{k-1} \right \}, \\
    \Sigma_{k \mid k-1} &= \E \left \{ [x_k-x_{k \mid k-1}][x_k-x_{k \mid k-1}]^\top \mid Z_{k-1},u_{k-1} \right \}.
\end{align*} 
At each time step, the eKF is the recursion
\begin{align}
 x_{k+1 \mid k} = f_\theta( x_{k \mid k},u_k),\quad \label{eq:eKF_state}
\Sigma_{k+1 \mid k} = F_k \Sigma_{k \mid k} F_k^\top + W, 
\end{align}
where
\begin{equation}
\begin{aligned} \label{eq:eKF_supportVariables}
& x_{k\mid k}= x_{k \mid k-1} + \Omega_{k} \left [y_{k}-h_\theta(x_{k \mid k-1}, u_k)\right],\\
&\Sigma_{k \mid k} = \left [I - \Omega_{k} H_{k} \right]\Sigma_{k \mid k-1},\\
 &\Omega_{k}= \Sigma_{k \mid k-1} H_{k}^\top \left [ H_{k} \Sigma_{k \mid k-1} H_{k}^\top+V\right]^{-1}, \\
 &F_k = \left. \frac{\partial f_\theta(x,u_k)}{\partial x} \right | _{x_{k\mid k}},\quad H_{k} = \left. \frac{\partial h_\theta(x,u_k)}{\partial x} \right | _{x_{k\mid k-1}}.
\end{aligned}
\end{equation}
The recursion is initialized by $ x_{0 \mid -1}$, $\Sigma_{0 \mid -1}$, the moments of $p_0$.

Note that the evolution of the eKF in \eqref{eq:eKF_supportVariables} depends on $y_k$. As an approximation, we drop this dependence, so we can approximate the evolution of the eKF to the future deterministically, before future measurements are received. This approximation can be plausible in many applications, since the measurement correction term $[y_{k} - h_\theta(x_{k\mid k-1}, u_k)]$, in \eqref{eq:eKF_supportVariables}, can be close to independent from $\Omega_{k}$, almost of a zero-mean white noise sequence when the eKF is fine-tuned \cite[Sec.~8.2]{anderson2012optimal}, and when $f_\theta$ is not highly nonlinear in its first argument. This weakens the need for the expectation in the finite-horizon $J_{\mathrm{sc}}$, and leads to a measurement-free version of the eKF in \eqref{eq:eKF_state} and \eqref{eq:eKF_supportVariables}. Now we can define the deterministic version of the eKF, propagating the information state $(x_{k \mid k}, \Sigma_{k \mid k})$,
\begin{equation*}
\mathcal{X}_{k+1} = 
\begin{pmatrix}
    x_{k+1 \mid k+1} \\
    \operatorname{upvec}\Sigma_{k+1 \mid k+1}
\end{pmatrix}
     = \mathcal{F}_\theta(\mathcal{X}_k, u_k),
\end{equation*}
where $\mathcal{F}_\theta$ is constructed by \eqref{eq:eKF_state} and \eqref{eq:eKF_supportVariables} when the term $y_k - h_\theta(x_{k \mid k-1},u_k)$ is set to zero, and $\operatorname{upvec}$ is the upper-triangle elements vectorized (to reduce the state dimension using symmetry).

\section{Proof of the LQG regularity Proposition}\label{app:lqgreg}

This appendix proves Proposition~\ref{prop:lqgreg} by verifying
Assumption~\ref{as:env}(A.I)-(A.IV) for the control SDP~\eqref{eq:Jcont}. The estimation
SDP~\eqref{eq:Jest} is its formal dual (swap $A_\theta\!\leftrightarrow
A_\theta^\top$, $B_\theta\!\leftrightarrow C_\theta^\top$, $W\!\leftrightarrow Q$,
$R\!\leftrightarrow V$; detectability for stabilizability); 

Write the program as $J^\star(\theta)=\min_z f(z,\theta)$ s.t.\
$z\in\mathcal C(\theta)$, with $z=(\Sigma,L,Z_0)$, linear objective
$f=\operatorname{tr}(Q\Sigma)+\operatorname{tr}(RZ_0)$, and $\mathcal C(\theta)$ the
two inequalities $M_1=\left[\begin{smallmatrix}Z_0&L\\ L^\top&\Sigma
\end{smallmatrix}\right]\succeq0$ and $M_2=\left[\begin{smallmatrix}\Sigma-W&
A_\theta\Sigma+B_\theta L\\ \star&\Sigma\end{smallmatrix}\right]\succeq0$. Assume
$(A_\theta,B_\theta)$ stabilizable and $Q,R,W\succ0$.

\emph{(A.I) Convexity, smooth data, strong duality.} $f$ is linear and
$\mathcal C(\theta)$ is defined by inequalities affine in $z$ (hence a convex set)
whose coefficients are affine in the model matrices $(A_\theta,B_\theta,Q,R,W)$, so
the data are jointly $C^1$ in $(z,\theta)$. A strictly feasible point exists: for any
stabilizing $K_0$ let $\Sigma_0\succ0$ solve $A_{\mathrm{cl},0}\Sigma_0
A_{\mathrm{cl},0}^\top-\Sigma_0+W\prec0$ ($A_{\mathrm{cl},0}=A_\theta+B_\theta K_0$
Schur, $W\succ0$) and set $L_0=K_0\Sigma_0$. The LMI slack variable $Z_0$ can be chosen freely such that $Z_0\succ K_0\Sigma_0K_0^\top$. Slater then gives strong duality with attained primal
and dual solutions \cite{boyd2004convex}.

\emph{(A.II) Unique primal minimizer.} The uniqueness of $(K^\star, \Sigma^\star)$ is a well-established fact in the LQG literature (given the positive definiteness of the weight matrices and controllability/observability of the system).

\emph{(A.IV) Smooth parameter entry.} $\theta$ enters only through $A_\theta,B_\theta$ (and, if co-designed,
$W,Q,R$), each assumed $C^1$ in $\theta$.

\emph{(A.III) Unique dual, strict complementarity.} We prove it by constructing the dual and showing it is a singleton. Using the
Lyapunov equality $\Sigma^\star-W=A_{\mathrm{cl}}\Sigma^\star A_{\mathrm{cl}}^\top$
and $A_\theta\Sigma^\star+B_\theta L^\star=A_{\mathrm{cl}}\Sigma^\star$,
\begin{equation}\label{eq:rank1}
M_1^\star=\begin{bmatrix}K^\star\\ I\end{bmatrix}\Sigma^\star
\begin{bmatrix}K^\star\\ I\end{bmatrix}^{\!\top}\!,\quad
M_2^\star=\begin{bmatrix}A_{\mathrm{cl}}\\ I\end{bmatrix}\Sigma^\star
\begin{bmatrix}A_{\mathrm{cl}}\\ I\end{bmatrix}^{\!\top}\!,
\end{equation}
each of rank $n$, with $\ker M_1^\star=\operatorname{range}[\,I;-K^{\star\top}]$ and
$\ker M_2^\star=\operatorname{range}[\,I;-A_{\mathrm{cl}}^\top]$. For multipliers
$N\succeq0$ (of $M_1$) and $S\succeq0$ (of $M_2$), complementary slackness
$NM_1^\star=M_1^\star N=0$, $SM_2^\star=M_2^\star S=0$ force $\operatorname{range}N\subseteq\ker M_1^\star$
and $\operatorname{range}S\subseteq\ker M_2^\star$; a positive semi-definite matrix with range in
$\operatorname{range}(\mathcal{X})$ must factor as $\mathcal{X} \mathcal{Y} \mathcal{X}$, $\mathcal{Y} \succeq 0$, so
\begin{equation}\label{eq:lift}
\begin{aligned}
&N=\begin{bmatrix}I\\ -K^{\star\top}\end{bmatrix}\Theta
\begin{bmatrix}I\\ -K^{\star\top}\end{bmatrix}^{\!\top}\!,\quad
S=\begin{bmatrix}I\\ -A_{\mathrm{cl}}^\top\end{bmatrix}\Xi
\begin{bmatrix}I\\ -A_{\mathrm{cl}}^\top\end{bmatrix}^{\!\top}\!,\\
&\Theta \succeq 0 , \Xi \succeq 0.
\end{aligned}
\end{equation}
Thus the whole dual is carried by two cores $\Theta=N_{11}$, $\Xi=S_{11}$, with
$N_{12}=-\Theta K^\star$ and $S_{12}=-\Xi A_{\mathrm{cl}}$. The Lagrangian
$ L=\operatorname{tr}(Q\Sigma)+\operatorname{tr}(RZ_0)-\langle N,M_1\rangle
-\langle S,M_2\rangle$ has three stationarity conditions,
\begin{equation}\label{eq:kkt}
\begin{aligned}
&\partial_{Z_0}\!:\,\Theta=R,\quad
\partial_{L}\!:\,N_{12}=-B_\theta^\top S_{12},\\
&\partial_{\Sigma}\!:\,Q-N_{22}-S_{11}-S_{22}-(A_\theta^\top S_{12}+S_{12}^\top A_\theta)=0.
\end{aligned}
\end{equation}
The first fixes $\Theta=R$, hence $N_{12}=-RK^\star$ and $N_{22}=K^{\star\top}RK^\star$;
the second then reads
\begin{equation}\label{eq:star}
B_\theta^\top\Xi A_{\mathrm{cl}}=-RK^\star .
\end{equation}
Writing $A_\theta=A_{\mathrm{cl}}-B_\theta K^\star$ in $\partial_\Sigma$, the
relation \eqref{eq:star} cancels every $B_\theta$-cross term (explicitly
$A_\theta^\top\Xi A_{\mathrm{cl}}+A_{\mathrm{cl}}^\top\Xi A_\theta=
2A_{\mathrm{cl}}^\top\Xi A_{\mathrm{cl}}+2K^{\star\top}RK^\star$) and
$\partial_\Sigma$ collapses to the closed-loop Stein equation
\begin{equation}\label{eq:stein}
\Xi=Q+K^{\star\top}RK^\star+A_{\mathrm{cl}}^\top\Xi A_{\mathrm{cl}} .
\end{equation}
Since $A_{\mathrm{cl}}$ is Schur, \eqref{eq:stein} has a unique solution; the Lyapunov equation in the observability gramian is $P=Q+K^{\star\top}RK^\star+A_{\mathrm{cl}}^\top PA_{\mathrm{cl}}$, so
$\Xi=P$. Hence the dual is the single point $S=[\,I;-A_{\mathrm{cl}}^\top]P[\,I;
-A_{\mathrm{cl}}^\top]^\top$ with $S_{11}^\star=P$ (and $N = [\,I; -K^{\star\top}] R [\,I;
-K^{\star\top}]^\top$). Strict complementarity follows by rank counting:
$\operatorname{rank}N=m$ ($\Theta=R\succ0$) and $\operatorname{rank}S=n$
($\Xi=P\succ0$ from $Q\succ0$) complete the ranks of $M_1^\star,M_2^\star$ (both $n$)
to full.
\hfill$\square$

\section{\texorpdfstring{Existence of a subgradient under (A.I) and a unique primal or dual solution}{Existence of a subgradient under (A.I) and a unique primal or dual solution}}\label{app:subgrad}
 
\providecommand{\Lag}{\mathfrak{L}}
\providecommand{\Zs}{\mathcal{Z}^{\star}}
\providecommand{\Ls}{\Lambda^{\star}}
\providecommand{\tr}{\operatorname{tr}}

Let $\lambda=(\mu,S)$ with $S\succeq0$, and
$\mathfrak L=J+\langle\mu,g\rangle-\langle S,G\rangle$.
$\Zs(\theta)$ and $\Ls(\theta)$ denote the primal and dual optimal sets, $(z^\star,\lambda^\star)\in\Zs(\theta)\times\Ls(\theta)$ is the pair returned by the solver, and $\hat g:=\nabla_\theta \mathfrak L(z^\star,\lambda^\star,\theta)$.
We write $\nabla_\theta$ for the partial gradient and keep $\partial$ for
the Clarke subdifferential \cite[Sec.~2.1]{clarke1983optimization}.
 
Note that $\Zs(\theta')$ and $\Ls(\theta')$ are nonempty and lie in a
fixed compact set $K$ for all $\theta'\in\Theta_0$, $\Theta_0$ a closed ball around $\theta$: for the primal sets this is the boundedness assumed in Proposition~\ref{prop:subgrad}, and for the dual sets it follows from Slater's condition. This can be seen if we let $\hat z$
be a Slater point (strictly feasible point) at $\theta$. That is, shrinking $\Theta_0$ if necessary, $G(\hat z,\theta')\succeq\sigma I$
for some $\sigma>0$ and all $\theta'\in\Theta_0$. By continuity of $J$ in $\theta$, every optimal point there costs at most $\alpha=\max_{\theta' \in \Theta_0}J(\hat z,\theta')$.
In the primal case, e.g.\ for the control SDP~\eqref{eq:Jcont}: (i) $\Sigma,Z_0\succeq0$, (ii) $L \Sigma^{-1} L^\top \preceq Z_0$, and by the above bound (iii) $\tr(Q\Sigma)+\tr(RZ_0)\le\alpha$
with $Q,R\succ0$, (i)-(iii) bound the primal $(\Sigma,L,Z_0)$. In the dual variables case: for $\lambda=(\mu,S)\in\Ls(\theta')$, strong duality gives
$J^\star(\theta')\le \mathfrak  L(\hat z,\lambda,\theta')=J(\hat z,\theta')
-\langle S,G(\hat z,\theta')\rangle\le J(\hat z,\theta')-\sigma\tr S$,
so $\|S\|_F\le\tr S\le(\alpha-\inf_{\theta'}J^\star(\theta'))/\sigma$, where $\inf_{\theta'}J^\star(\theta')\ge\min_{K\times\Theta_0}J>-\infty$ by the primal bound
(cf.\ \cite[Thm.~4.2]{bonnans1998guided}).

For $z_2\in\Zs(\theta_2)$ and
$\lambda_1\in\Ls(\theta_1)$,
\begin{equation}\label{eq:C1}
J^\star(\theta_2)-J^\star(\theta_1)\ \ge\
\mathfrak L(z_2,\lambda_1,\theta_2)-\mathfrak L(z_2,\lambda_1,\theta_1).
\end{equation}
This holds because $J^\star(\theta_2)=J(z_2,\theta_2)\ge
\mathfrak L(z_2,\lambda_1,\theta_2)$ by weak duality ($z_2$ is feasible at $\theta_2$ and $S_1\succeq0$), while
$J^\star(\theta_1)=\inf_z \mathfrak L(z,\lambda_1,\theta_1)$ by strong duality. By continuity of $\nabla_\theta \mathfrak L$, there exists $M=\max_{K\times\Theta_0}\|\nabla_\theta \mathfrak L\|$. Applying \eqref{eq:C1}
both ways with the mean value theorem on the segment $[\theta_1,\theta_2]\subset\Theta_0$ gives $|J^\star(\theta_2)-J^\star(\theta_1)|\le
M\|\theta_2-\theta_1\|$. By continuity of $J^\star$ and $L$, limit
points of optimal solutions at $\theta+td$ are optimal at $\theta$
(cf.\ \cite[Sec.~4.1]{bonnans2000perturbation}). Hence, as $t\downarrow0$,
any $z_t\in\Zs(\theta+td)$ converges to $z^\star$ if the primal solution is
unique, and any $\lambda_t\in\Ls(\theta+td)$ converges to $\lambda^\star$
if the dual solution is unique.
 
Take \eqref{eq:C1} with $(\theta_2,\theta_1)=
(\theta+td,\theta)$, $z_t$ and $\lambda^\star$, and apply the mean value
theorem in $\theta$:
\[
\begin{aligned}
J^\star(\theta+td)-J^\star(\theta)
&\ge t\langle\nabla_\theta \mathfrak L(z_t,\lambda^\star,\theta+ s_t\, t d),d\rangle\\
&=t\langle \hat g,d\rangle+o(t)
\end{aligned}
\]
for some $s_t\in(0,1)$, when $z_t\to z^\star$, i.e.\ when the primal solution is unique. Exchanging the roles, i.e.\
$(\theta_2,\theta_1)=(\theta,\theta+td)$ with $z^\star$ and $\lambda_t$,
gives the reverse bound $J^\star(\theta+td)-J^\star(\theta)\le t\langle\hat g,d\rangle+o(t)$ when $\lambda_t\to\lambda^\star$, i.e.\ when the dual solution is unique. Under both, $(J^\star)'(\theta;d)=\langle
\hat g,d\rangle$ for all $d$; a Lipschitz function with this property is
differentiable at $\theta$.
 
Now we get to the Clarke subgradient, when only one of the two solutions is unique. If the primal is unique, the first bound gives $(J^\star)^\circ(\theta;d)
\ge\limsup_{t\downarrow0}[J^\star(\theta+td)-J^\star(\theta)]/t\ge
\langle \hat g,d\rangle$. If the dual is unique, the reverse bound with $d$ replaced by $-d$ reads $[J^\star(\theta)-J^\star(\theta-td)]/t\ge\langle\hat g,d\rangle+o(1)$, the same estimate with base point $\theta-td\to\theta$, so again $(J^\star)^\circ(\theta;d)\ge\langle\hat g,d\rangle$. In either case, as $d$ is arbitrary, $\hat g\in\partial J^\star(\theta)$ by definition
(see \cite[Sec.~2.1]{clarke1983optimization} for more details). In the unique-dual case, the reverse bound with $d=-\hat g$ also gives $J^\star(\theta-t\hat g)\le J^\star(\theta)-t\|\hat g\|^2+o(t)$, so $-\hat g$ is a descent direction whenever $\hat g\neq0$; with only the primal unique, this need not hold. If neither solution is unique, $\hat g$ need not be a subgradient: for $\min_{z,s}\{s:\ s\ge\pm\theta z,\ |z|\le1\}$ (Slater holds, data affine in $\theta$), $J^\star\equiv0$, yet at $\theta=0$ the KKT pair $z=1$, $s=0$ with multipliers $(1,0)$ on $s\ge\pm\theta z$ gives $\hat g=1\notin\{0\}=\partial J^\star(0)$.

\bibliographystyle{plain}
\bibliography{References}

\vspace{0.1cm}
\begin{flushright}
	\scriptsize \framebox{\parbox{2.5in}{Government License: The
			submitted manuscript has been created by UChicago Argonne,
			LLC, Operator of Argonne National Laboratory (``Argonne").
			Argonne, a U.S. Department of Energy Office of Science
			laboratory, is operated under Contract
			No. DE-AC02-06CH11357.  The U.S. Government retains for
			itself, and others acting on its behalf, a paid-up
			nonexclusive, irrevocable worldwide license in said
			article to reproduce, prepare derivative works, distribute
			copies to the public, and perform publicly and display
			publicly, by or on behalf of the Government. The Department of Energy will provide public access to these results of federally sponsored research in accordance with the DOE Public Access Plan. http://energy.gov/downloads/doe-public-access-plan. }}
	\normalsize
\end{flushright}	

\end{document}